\documentclass[12pt]{article}
\usepackage{hyperref}
\usepackage[sort,compress]{cite}

\usepackage{amsmath}
\usepackage{amsthm}
\usepackage{amssymb}
\usepackage{amsfonts}
\newtheorem{theorem}{Theorem}
\newtheorem{definition}{Definition}
\newtheorem{remark}{Remark}

\title{On Topology Changes in Quantum Field Theory and Quantum Gravity}
\author{Benjamin Schulz}

\vspace{10pt}
\begin{document}
\maketitle

\begin{abstract}
	Two singularity theorems can be proven if one attempts to let a Lorentzian cobordism  interpolate between two topologically distinct manifolds. On the other hand, Cartier and DeWitt-Morette have given a rigorous definition for quantum field theories (qfts) by means of path integrals. This article uses their results to study whether qfts can be made compatible with topology changes. We show that path integrals over metrics need a finite norm for the latter and for degenerate metrics, this problem can sometimes be resolved with tetrads. We prove that already in the neighborhood of some cuspidal singularities, difficulties can arise to define certain qfts. On the other hand, we show that simple qfts can be defined around conical singularities that result from a topology change in a simple setup. We argue that the ground state of many theories of quantum gravity will imply a small cosmological constant and, during the expansion of the universe, will cause frequent topology changes. Unfortunately, it is difficult to describe the transition amplitudes consistently due to the aforementioned problems. We argue that one needs to describe qfts by stochastic differential equations, and in the case of gravity, by Regge calculus in order to resolve this problem.
\end{abstract}

%
\vspace{2pc}

%
%

%
%

\section{Introduction}\label{s0}
In this work, we investigate the compatibility of quantum field theories and various proposals for quantum gravity with topological changes of the space-time. 

For this, one needs at first precise mathematical definitions of the space-time topology and of quantum field theory, which we will give in this introduction. From the mathematical analysis, it is known that a topology change of a space-time will lead either to closed time-like curves or to singularities. The mathematics of quantum field theory and quantum gravity in this article will be based on the path integral formalism of DeWitt-Morette and Cartier. 

The introduction can be seen mostly as a review. Although the author has not seen the formalism of DeWitt-Morette applied on quantum field theory and quantum gravity in such a way. (Among the new results in the introduction is that the famous problem of time in quantum gravity goes simply away, once one writes the gravitational path integral down in a careful way that emphasizes the general role of the space-time over which one integrates as a cobordism that always has boundaries).

The earliest article noting the incompatibility of 2 dimensional quantum field theory with a  $ 2 $  dimensional topology change involving a conical singularity was given by Anderson and DeWitt \cite{Topology}.

An action in quantum field theory usually involves an integral of squares of derivatives of the paths.
For example, the action the so-called  $ \varphi^4 $  theory on a Minkowski space  $ M $  with metric tensor  $ \eta $  is given by
\begin{equation}
	S(\varphi(x))=\int_M d^4x\frac{1}{2}\eta^{\mu\nu}\partial_\mu\varphi\partial_\nu\varphi-\frac{m}{2}\varphi^2-\frac{\lambda}{4!}\varphi^4.
\end{equation}

Since any integral (the path integral included), over an integrand exists of course only if the latter can be defined, one has to conclude that for such an action, the space of paths is restricted to functions   $ \varphi(x) $  whose derivatives are square integrable.

Because of this issue, the author suggested in his phd thesis \cite{phd} that topology changes were generally forbidden if topological transitions over space-times with closed time-like curves were excluded.

However, one could assume that different action functionals, e.g ones where derivatives do not occur, could solve the problem. In \cite{DeWitt9} and \cite{DeWittbook}, DeWitt has given a comment in which he outlines an argument why homotopy considerations would forbid topology changes. This argument solely depends on the measure of the path integral and is independent from an action. Unfortunately, DeWitt's comment is just a short sentence from which it is not immediately clear to which kinds of topology changes and to which quantum field theories it does apply.

In this article, we will expand this comment in section \ref{s1} into a result for path integrals over metrics. In this article, we will argue that DeWitt's argument can mostly be resolved by the solution of Horowitz \cite{Horowitz,Horowitz2,Gidd}.

The case of fields on a singular space-time with conical singularities also needs closer look. String theory was written on the vicinity of conical singularities of orbifolds. In \cite{DeWitt9},  DeWitt asked whether this could perhaps be used to describe topology changes with a quantum field theory, but his comment does not give a concluding answer.

Furthermore, Borde,  Dowker, Garcia, Sorkin and Surya have argued in \cite{Dowker} that certain topology changes involving singularities where the metric vanishes could be allowed. These authors write that some physical criteria called causality protection should be fulfilled. They have, however, not used a rigorous mathematical definition of a quantum field theory to study whether their constructions are compatible with quantum fields or not.  

Using our mathematical notion of quantum field theory, we are able to analyze the question of whether and how quantum field theories can be made compatible with singularities in a mathematically more precise way. This will be done in section \ref{s2} and \ref{s3}. 

As we will see, the required  consistency checks involve the neighborhood of the singularity. More precisely, we argue that the study of Cauchy problems for fields close to the singularity is necessary. We prove that for some cuspidal singularities, certain classical and quantum field theories can not be consistently defined even in the neighborhood of the singularity once the theory has been extended to the latter. 

Many claims about the possibility of topology changes were made in the string theory literature. We will discuss some these claims in section \ref{sa} and what they specifically entail. Unfortunately, we show them often to be very exaggerated.

In section \ref{s3}, we will discuss a simple example of a topology change involving a conical singularity that we believe is compatible with quantum field theory for fields on the singular space-time. We use this working  example and contrast it to the singular trousers problem of Anderson and DeWitt from \cite{Topology} in order to show that not only the singularity can be a source of problems for a quantum field theory but that a topology change can cause difficulties for the global analysis of the quantum field.

In section \ref{s4}, we will  turn to the ground state of quantum gravity. We will review arguments from Hawking \cite{foam} which show that if quantum gravity is quantized as a correction to a classical background field, then one has to integrate the path integral also over the classical backgrounds. Hawking's argument was criticized by Christensen and Duff\cite{chris}, who noted that Hawking's amplitude may not converge. We will argue that Hawking's amplitude converges if one sets the regularization parameter correctly. We review arguments from the author's phd thesis, which show that this then leads to a ground state with a small cosmological constant \cite{phd}. From \cite{foam}, one gets the result that the Euler characteristic of the space-time is proportional to the 4-volume of the space-time. In an expanding universe, this would imply frequent topology changes. We will argue that similar conclusions hold for other theories of of quantum gravity.

In the remainder of this article, we will look what mathematical properties a theory must have that it can be more compatible with the many difficulties arising from topological changes of the space-time. An unexpected hint comes from Bell's theorem. In section \ref{s5}, we review Nelson's\cite{Ne4} and Faris' \cite{Far} rigorous analysis of Bell's theorem. It shows that any theory is excluded where the outcomes of EPR experiments are determined by events that happen before the measurement stage. We will argue that this implies that any valid theory of nature must include probabilistic elements. 

We will argue in section \ref{s6} that one may use stochastic processes to model relativistic quantum field theories in a way that they are compatible with the the experimentally observed violations of Bell's inequality. Stochastic processes  have non-differentiable paths. We argue that this property makes them more compatible with topological changes. It opens the possibility to describe the singular trousers model of Anderson and DeWitt from \cite{Topology} without divergences in the energy momentum tensor and it resolves the problems that we will prove to arise in the case of cuspidal singularities. 

Finally, we conclude the article with section \ref{s7} by mentioning that a  discretization which is similar to those in stochastic processes can be modeled with Regge calculus in general relativity. We discuss how this may be used to write path integrals over metrics and with the Einstein action as integrand. The singularities in this formalism are conical, which means that the aforementioned difficulties with cuspidal singularities are not present. 

All this is certainly not exhaustive. The characterization of singularities in general relativity is difficult. Even with a concise description of the path integral, one would need a much more expanded treatise to fully characterize the singularities with which quantum field theories are compatible.

From what is written above, it follows that on the mathematical side, the reader should have some knowledge of analysis, functional analysis, partial differential equations, topology, homotopy theory, differential geometry,  algebraic geometry and singularity resolution, see e.g. \cite{sing}, functional integration, see \cite{functionalintegration}, probability theory and stochastic processes, see\cite{Bauer}. On the physical side, one should have good knowledge on quantum mechanics, see \cite{Auletta}, quantum field theory in flat and in curved-space-times, see e.g. \cite{DeWittbook1, DeWittbook,DeWittcurved,DeWitt50years}, constructive quantum field theory, see \cite{Wightman,schottenloher}, general relativity, quantum gravity \cite{DeWittbook1, DeWitt, DeWitt22, DeWitt3}, Euclidean quantum gravity \cite{Hawkingbook} and, finally, on mathematical string theory along the lines of \cite{deligne}. 

Furthermore, the reader should be warned that this article does not contain many aspects of physical model building. Neither does it make many predictions. Instead, it starts from rather broad assumptions about quantum gravity and then comes to certain conclusions about broad aspects that physical theories must have by means of mathematical deductions.  Model building as it is often done by physicists would usually include many assumptions and hypotheses that one can not prove and then go into detailed predictions. This article does not aim to do that. Instead, it starts from very general assumptions that hold for large classes for models of quantum gravity theories and then makes rather broad conclusions that mathematically follow from them.

\subsection{Space-time topology}

We define a topology change as the change of a manifold  $ S_1 $  into a different  manifold  $ S_2 $  that is not diffeomorphic to  $ S_1 $. Naively, one may ask whether this transition can by done smoothly and in a way that is respecting physical causality in some way. In this case, there would be some kind of interpolating manifold, or cobordism between  $ S_1 $  and  $ S_2 $. In the following, we make these ideas more precise. In order to attempt this, we need some of the definitions of the causal structure of a space-time  from the relativity literature, e.g. \cite{causalstructure,Geroch0, HawkingEllis}.

\begin{definition}	Let  $ (M,g) $  be a manifold with Lorentzian metric  $ g $. Then  $ (M,g) $  is called a Lorentzian manifold. A tangent vector X is called time-like if  $ g(X,X)<0 $, light-like if  $ g(X,X)=0 $  and space-like if  $ g(X,X)>0 $.
	
	Let  $x^\mu(\tau)\in M$ be a curve parameterized by a curve parameter  $\tau\in U$  where  $U$  is some some interval  $U\subseteq \mathbb{R}$, and let   $ \mu\in\{1,\ldots,n\}$  be the indices of the components of each point from the curve.  $ x^\mu(\tau) $  is a  $ C^1 $  curve if the components of the tangent  $ \frac{dx^\mu}{d\tau}$  exist and are continuous.  A curve   $ x^\mu(\tau) $  is called time-like, if it is not just defined at a single point and if is a  $C^1$  curve with a tangent that is everywhere time-like.
	
	$(M,g)$ is chronological, if it contains no closed time-like curves. A Lorentzian metric $g$  that allows a global choice of future and past for any point $p\in M$ is called time-orientable, and if that choice was made, time-oriented. A space-time is a time oriented Lorentzian manifold of at least dimension $2$.
	
	A  $C^1$  curve  $x^\mu(\tau)$  is  future/past-directed if the tangent vectors point in the future/past time direction. Given a point  $ p\in M $, its chronological future/past is defined as the set  \begin{equation} I^{\pm}(p)=\{q|\exists \text{ a future/past directed time-like curve from p to q}\} \end{equation}  and for a space-time which is time orientable, we can define  $ I(p)=I^+(p)\cup I^-(p) $. 
\end{definition}

In order to describe the topology change mathematically on a space-time with such causality structures we now need some of the definitions from  the very good introduction of Borde \cite{Borde} and Larssen \cite{Larsson}. 

We mention these results and proofs here, even though they are known in the mathematics community, because cobordisms have recently attracted attention of the string theory community in form of a so-called cobordism hypothesis \cite{cumrun}. Therefore, these results maybe of use to other researchers in the field.

\begin{definition}Let  $ S_1 $  and  $ S_2 $  be  manifolds without boundary and of dimension  $ n $, where  $ n\in\mathbb{N} $. We say that there is a pseudo-cobordism  $ M $  between them if  $ M $  is a connected  $ n+1 $  dimensional  manifold whose boundary is given by  $ \partial M=S_1 \sqcup S_2 $, where  $ \sqcup $  denotes the disjoint union, i.e. one has   $ \partial M=S_1 \cup S_2 $  and  $ S_1 \cap S_2 =\emptyset $. If the pseudo-cobordism  $ M $  is compact, it is called cobordism.
	
	Let  $ (M,g) $  be a time-oriented pseudo cobordism with a Lorentzian metric  $ g $. Then,  $ (M,g) $  is called Lorentzian pseudo-cobordism. A Lorentzian pseudo-cobordism  $ (M,g) $  is called a weak Lorentzian pseudo-cobordism, if  $ S_1 $  and  $ S_2 $  are space-like hyper-surfaces with respect to the Lorentz structure on  $ (M,g) $. If a weak Lorentzian pseudo-cobordism is compact, it is called a weak Lorentzian cobordism. If the closure  $ \overline{I(p)} $  of a (weak) Lorentzian pseudo-cobordism  $ (M,g) $  is compact, we call  $ (M,g) $  a causally compact (weak) Lorentzian cobordism.
	\label{defn2}
\end{definition}
Let us assume that the manifolds  $ S_1 $  and  $ S_2 $  are not diffeomorphic, and so there happens a topology change between them. The first question is then  whether a cobordism exists between them. For this case, if  $ S_1 $  and  $ S_2 $  are compact, one has a standard result in differential geometry:
\begin{theorem}(Milnor)Let  $ S_1 $  and  $ S_2 $  be two n-dimensional compact manifolds without boundary. There is a cobordism between them if and only if they have the same Stiefel-Whitney numbers.
\end{theorem}
\begin{proof}
	For a proof, see \cite{Milnor}.
\end{proof} This is not a severe restriction. E.g. any two compact manifolds without boundary have a cobordism between them if  $ n=3 $.

As a next step, one may ask whether there can be a Lorentzian cobordism. For compact  $ S_1,S_2 $, there is a theorem by Reinhart and Misner, which states:

\begin{theorem}(Reinhart, Misner) Let  $ S_1 $  and  $ S_2 $  be compact manifolds without boundary. Then there is a Lorentzian cobordism between them if and only if they have the same Stiefel-Whitney numbers	and Euler numbers.
\end{theorem}
\begin{proof}For a proof, see \cite{Reinhart, Borde, Larsson}.
\end{proof}

Finally we want the cobordism to be chronological. And we want to relax the compactness assumptions of  $ S_1,S_2 $.  In 1967, Geroch has given an important proof on the conditions for the existence of a chronological, weak Lorentzian cobordism if  $ S_1, S_2 $  are not compact and of dimension  $ 3 $, see \cite{Geroch}. In his work, Geroch already suggested to relax the assumptions of his theorem by the use of covering manifolds.

In recent years, his theorem was somewhat extended. Now  $ S_1 $  and  $ S_2 $  can have n-dimensions, where  $ n\in\mathbb{N} $, and the compactness condition of the chronological weak Lorentzian cobordism was relaxed to the causal compactness. In this modern form given by \cite{Borde, Larsson}, the theorem goes as follows:

\begin{theorem}(Geroch, Borde)Let  $ S_1 $,  $ S_2 $  be n- dimensional manifolds without boundary (not necessarily compact and not necessarily connected). Let  $ (M, g) $  be a chronological, weak and causally compact Lorentzian cobordism between  $ S_1 $  and  $ S_2 $. Then there is a diffeomorphism  $ \varphi: S_1 \times [0, 1] \mapsto M $  such that the sub-manifold
	
	$ \varphi(\left\{x\right\}\times[0, 1]) $  is time-like for every  $ x\in S_1 $. Furthermore,  $ (M,g) $  is topologically trivial (i.e. diffeomorphic to $S_1 \times [0, 1]$), and  $ S_1 $  and  $ S_2 $  are diffeomorphic.
\end{theorem}
\begin{proof}For a proof, see \cite{Borde,Larsson}.
\end{proof}

This result implies that if  $ S_1 $  and  $ S_2 $  are not diffeomorphic, one can not find a chronological,  weak and causally compact Lorentzian cobordism between them. 

For the reader who is new to these definitions, it should be emphasized that it is in general no problem to find cobordisms and even Lorentzian cobordisms between  $ S_1 $  and  $ S_2 $. For example, Yodzis constructed Lorentzian cobordisms with Morse surgery\cite{Yodzis}. However, once he arrives at his non-singular construction, the Cobordism then has to violate the other requirements of Geroch's theorem, in Yodzis' case they have closed time-like curves. 

So, for the reader who is not well versed in this material, let us see what all the definitions in Geroch's theorem entail exactly.  It is clear that we have to use a causally compact space-time as cobordism, because we want the space-time to include the manifolds  $ S_1, S_2 $  as boundaries in order to complete the topological transition between them. The condition of a chronological space-time (no closed time-like curves) is common in the physical literature and necessary in order to avoid several paradoxes. We therefore want to shortly emphasize again what the condition of a "weak Lorentzian" cobordism exactly implies.

\begin{theorem}Let  $ M $  be a manifold of dimension  $ n $,  $ n\in \mathbb{N} $.  $ M $  can be given a Lorentzian metric  $ g $  such that  $ (M,g) $  is time orientable if and only if  $ M $  admits globally a non-vanishing vector field.\label{theoremvector}
\end{theorem}
\begin{proof}See \cite{Larsson} for both directions of the proof. Here we just construct the metric from the vector field since we will use the form of the metric below. Let  $ v $  be a non-vanishing vector field on  $ M $. With  $ M $  being a  manifold, it can be given globally a Riemannian metric  $ h $. Let   $ v_\mu $,  $ h_{\mu\nu}  $  be the components of  $ v $  and  $ h $. Then,  $ g_{\mu\nu}=(h_{\alpha\beta}v^\alpha v^\beta)h_{\mu\nu}-2h_{\mu \alpha}v^\alpha h_{\nu \beta} v^{\beta} $  are the components of a Lorentzian metric tensor for  $ M $.  $ v $  is time-like with respect to the Lorentz structure of  $ g_{\mu\nu} $  and can be used to define the past and future at each point of  $ M $.
	\label{proofvectorfield} 
\end{proof}

If the Lorentzian cobordism in definition \ref{defn2} is weak, then  $ S_1,S_2 $  are space-like hyper-surfaces with respect to the Lorentz structure on  $ M $. Hence, for a weak cobordism, the vector field  $ v $  should, without loss of generality, be interior normal to  $ S_1 $  and exterior normal to  $ S_2 $. 

\begin{remark}Obviously, all these mathematical definitions ensure that one can speak of  $ S_1,S_2 $  as initial and final hyper-surfaces, and that one has a global designation of what constitutes the past and future of every point in the interpolating space-time. Hence, in Geroch's words, "a continuous choice of the forward light cone can be made", a property which he calls "isochronous". 
	
	We now assume we would want to go around Geroch's theorem. Lets say we violate the condition that the vector field  $ v $  is everywhere non-zero. In that case, the Lorentzian metric  $ g_{ab} $  would vanish, but we could possibly use the Riemannian metric  $ h $. To adopt  $ h $  as metric for the neighborhood of a point is physically problematic, since then,  causality structures would get lost for all points within this neighborhood. 
	
	Another way around Geroch's theorem would be not to use cobordisms. This would imply abandoning differentiable manifolds, and instead admit topological spaces which consist of manifolds and sets of singular points. Using topological spaces with singularities would have the advantage that the problems with causality could possibly be confined to isolated points.
\end{remark}
In order to recognize singularities, we now make the following definitions:
\begin{definition}Let  $ M=U \sqcup V $  be a topological space and  $ U $,  $ V $  disjoint subsets of  $ M $, where  $ U $  is a manifold, i.e. a space which is Hausdorff, second-countable, paracompact and where every point  $ p\in M $  has a neighborhood which is homeomorphic to  $ \mathbb{R}^n,n\in\mathbb{N} $. Let the points in  $ V $  do not have such a neighborhood. Following \cite{sing}, we call them singularities.  \end{definition}

But are models that have singularities physically feasible? Physicists use quantum mechanics, which was applied to fields and developed into quantum field theory in order to describe the microscopic world.

In order to study whether the occurrence of a singularity is compatible with quantum field theory, one first needs a rigorous framework of quantum field theory. One of these frameworks are certainly the Wightman axioms, see \cite{Wightman,schottenloher}. These require elements of functional analysis, especially tempered distributions, that make smooth structures necessary. However, it is often rather difficult to show whether a specific model satisfies all of these axioms. In order to work near or even at singularities, we need a definition of quantum field theory that is flexible and has the fewest restrictions that are possible. 

\subsection{Quantum Field Theory}
Physicists usually describe their quantum field theories in terms of path integrals.  The latter are functional integrals of the form  $ \int d\varphi(x,t)e^{s S(\varphi(x,t))} $, where   $ \varphi(x,t): \{\mathbb{R}^n, \mathbb{R}_+\}\rightarrow \mathbb{C} $  is a function which is called a field. The field is  element of a Banach space  $ \Phi $ and  $ S(\varphi):\Phi\mapsto \mathbb{R} $ is a functional of these fields called "action functional". Usually,  $ S(\varphi) $  also depends on partial derivatives of the fields. If  $ s=-1 $, one can describe the path integral in terms of ordinary measure theory, see \cite{functionalintegration}. However, physicists use complex integrals  $ s=i $  and then, one can show that  $ d\varphi(x,t) $  can not be understood with ordinary Lebesgue measure theory. 

As we will see, path integrals can be formulated rigorously and they need slightly less restrictive assumptions than the Wightman axioms. Furthermore, they lead the physicist straightforwardly to a model from which he can do calculations. So we first make  a rigorous definition of quantum field theories in terms of path integrals before we study them near and at singularities induced by topology changes.

A first rigorous definition of path integrals was provided by C. DeWitt-Morette in \cite{Morette4}. As noted in \cite{Cartier}, an integral  $ \int_a^b dx=b-a $, where  $ a,b,x\in \mathbb{R} $  can be used to define the integrator  $ dx $, which could also be defined by  $ \int_\mathbb{R} \exp{(-\pi x^2-2\pi i xy)}dx=\exp{(-\pi y^2)} $  for  $ x,y\in\mathbb{R} $. In the same way one can attempt to define an integrator rigorously for a path integral.

A solution  $ |\psi(\mathbf{x},t)\rangle $  of the Schr\"odinger equation \begin{equation} 
	i\frac{\partial}{\partial t}|\psi(\mathbf{x},t)\rangle=-\frac{1}{2m}\Delta|\psi(\mathbf{x},t)\rangle+V(x)|\psi(\mathbf{x},t)\rangle\label{Schroe}
\end{equation} where  $ t\in\mathbb{R_+} $, and  $ \mathbf{x}\in \mathbb{R}^n,n\in\mathbb{N} $, can be provided by a path integral. In  \cite{Morette4}, it was not clear for which potentials  $ V(\mathbf{x}) $  this integral would exist. A general class of functions were provided in \cite{Albeverio} by Albeverio and H{\o}egh-Krohn, so one has
\begin{theorem}(DeWitt-Morette, Albeverio, H{\o}egh-Krohn)
	A solution of Eq.  (\ref{Schroe}) equation is given by \begin{equation}|\psi(\mathbf{x}_b,t_b)\rangle=\int_{\mathcal{P}_{ab}} d\Gamma^W(\mathbf{x})\exp{\left(-i\int_{t_a}^{t_b}dt V(\mathbf{x}(t))\right)}\cdot |\phi(\mathbf{x})\rangle,\end{equation}
	where   $ \tilde{P}_{t_a, t_b} $  is the space of absolutely continuous paths  $ X_a^b:t\rightarrow\mathbb{R}^n,t\in [t_a,t_b],t_a<t_b\in\mathbb{R},X(t_a)=\mathbf{x}_a\in \mathbb{R}^n,X(t_b)=\mathbf{x}_b\in\mathbb{R}^n $. The measure  $ d\Gamma^W $  is defined over  $ \int_{\mathcal{P}_{ab}} d\Gamma^W=\int_{\mathcal{P}_{ab}}d\mathbf{x} \exp{(iS_0(\mathbf{x}))} $, with the quadratic form  $ S_0(\mathbf{x})=\frac{m}{2}\int_{t_a}^{t_b}dt |\dot{\mathbf{x}}(t)|^2 $  and where  $ V(\mathbf{x})=\int_{\mathbb{R}^n}\exp{(i\mathbf{\alpha} \mathbf{x})}d\mu(\mathbf x) $,  $ |\phi(\mathbf{x})\rangle=\int_{\mathbb{R}^n}\exp{(i\mathbf{\alpha} \mathbf{x})}d\nu(\mathbf{\alpha}):=|\psi(\mathbf{x},t_a)\rangle  $  are Fourier transforms of bounded measures  $ \mu(\mathbf{x}),\nu(\mathbf{\alpha}) $, with $\mathbf{\alpha}\in \mathbb{R}^n$.
\end{theorem}

In \cite{Cartier}, Cartier and DeWitt-Morette then extended the notion of their path integral to quantum field theory. One has the following 

\begin{theorem}(Cartier, DeWitt-Morette) Given two bounded maps  $ \Theta:(\Phi\times\Phi')\rightarrow\mathbb{C} $  and  $  Z:\Phi'\rightarrow \mathbb{C} $, where  $ \Phi,\Phi' $  are two Banach spaces related by a scalar product  $ \langle,\rangle:\Phi'\times\Phi\mapsto \mathbb{R}_{+} $  and  $ \varphi\in\Phi,J\in\Phi' $, one can define an integrator $\mathcal{D}_{\Theta,Z}$
	\begin{equation}
		\int_\Phi \Theta(\varphi,J)\mathcal{D}_{\Theta,Z}\varphi=Z(J)
	\end{equation}and a normed space  $ \mathcal{F}_{\Phi,Z} $ of functionals on $\Phi$ integrable by $\mathcal{D}_{\Theta,Z}$. \label{integrator}
\end{theorem}
\begin{proof}See \cite{Cartier,Morette3,Morette2,functionalintegration}. The subject has been extensively covered in books and journal articles that are available online. So we will not repeat the proofs here.
\end{proof}

One can use this result to define a quantum field theory along the lines of DeWitt and DeWitt-Morette \cite{BDeWittMorette}.
\begin{definition}Let   $ x\in M $  be a point in an n-dimensional Lorentzian space-time manifold and \begin{equation}\Theta(\varphi(x),J(x))=\mu(\varphi(x))\exp{(i(S(\varphi(x))-\langle J(x),\varphi(x)\rangle))},\end{equation}
	with  $ S(\varphi(x)) $  as the classical action of a system,  $ J\in \Phi' $  is called "source". Let  $ x_1,x_2\in M $  and let the fields  $ \varphi(x_1),\varphi(x_2)\in \Phi $  fulfill  \begin{equation}\left(\varphi(x_i),\varphi(x_j)\right)_P=G^{+ij}-G^{-ij}\label{commutatorpeierls},\end{equation}
	where  $ (A,B)_P $  is defined by the Peierl's bracket 
	\begin{equation}(A,B)_P=B\frac{\overleftarrow{\delta}}{\delta\varphi(x_i))} G^{-ij} \frac{\overrightarrow{\delta}}{\delta\varphi(x_j))}A-A\frac{\overleftarrow{\delta}}{\delta\varphi(x_i))} G^{-ij} \frac{\overrightarrow{\delta}}{\delta\varphi(x_j))}B,\end{equation}
	with   $ G^{\pm ij} $  as the advanced/retarded Green's function of  $ S'' $  in the expansion \begin{equation}S(\varphi_0+\varphi)=S(\varphi_0)+\frac{1}{2}S''(\varphi_0) \varphi  \varphi+\ldots\label{perturb},\end{equation} 
	where  $ \varphi_0 $  solves classical equations of motion that can be obtained by minimizing the variation of  $ S $  and hence  $ S'(\varphi_0)=0 $.
	
	Denote two sets of space-time points on an n-dimensional Lorentzian manifold  $ M $  by  $ a $  and  $ b $, where all points in  $ b $  lie in the future light cone of al points in  $ a $.

	Let  $ \langle out| $    and  $ | in\rangle $  be quantum mechanical (bosonic) state vectors in a Hilbert space  $ \mathcal{H} $, where Dirac notation was used. Associate  $ |in\rangle $   with the quantum state at the beginning of some time dependent dynamics at  $ a $  and   $ \langle out| $  with a quantum state at the end of this dynamical process at   $ b $.
	By this we mean that we can write   $ |out\rangle $  as a functional\begin{equation} \langle out|=\langle \Psi(\varphi(b))|=N\int_{\mathcal{P}_{ab}}D_{\Theta,Z}\varphi \; \Theta(\varphi,J)\langle in|\label{pathint2}.
	\end{equation}
	In Eq.  (\ref{pathint2}),   $ N $  is  a normalization constant, and  $ \mathcal{D}_{\Theta,Z} $  is chosen such that it is invariant under translation, or  $ \int_{\mathcal{P}_{ab}} \frac{\delta}{\delta\varphi}(\Theta(\varphi,J))\mathcal{D}_{\Theta,Z}\varphi=0 $  and  $  \mathcal{D}_{\Theta,Z}(M\varphi)=det(M)\mathcal{D}_{\Theta,Z}\varphi $.
	
	The path integral was mathematically at first only constructed for Paths in a Banach space of loops  $ X:t\rightarrow\mathbb{R},t\in [t_a,t_b],t_a,t_b\in\mathbb{R},X(t_a)=0,X(t_b)=0 $. An affine transformation then allowed paths 
	$ \tilde{P}_{t_a, t_b}: X_a^b:t\rightarrow\mathbb{R}^n,t\in [t_a,t_b],t_a,t_b\in\mathbb{R},X(t_a)=\mathbf{x}_a\in \mathbb{R}^n,X(t_b)=\mathbf{x}_b\in\mathbb{R}^n $, which are not in a Banach space, see \cite{functionalintegration}, p. 64. Finally it was extended to paths  $ \mathcal{P}_{ab} $  of fields
	$ \varphi(x)\rightarrow\mathbb{C}^n,x\in [a,b]\subset M $, where  $ M $  is a Manifold and  $ a=(t_a,\mathbf{a})\in M $  is in the past light-cone of  $ b=(t_b,\mathbf{b})\in M $  and  $ \varphi(a)=\varphi_a\in\mathbb{C}^n, \varphi(b)=\varphi_b\in \mathbb{C}^n $   are Cauchy data for the field, see \cite{functionalintegration}, p. 291. This allows us to use theorem \ref{integrator} for the existence of the integrator. 
	
	The expectation value of a field operator  $ \hat\varphi(x) $  acting on  $ \mathcal{H} $, with  $ x\in M $  being in the future light-cone of  $ a $  and the past light-cone of  $ b $   is given by \begin{equation}\langle out| \hat\varphi(x))|in\rangle=\frac{1}{i}\frac{\delta}{\delta J}\langle out|in\rangle.\label{expect}\end{equation}
	Let  $ x_k\in M  $, with  $ k $  as index set, be in the past light-cone of  $ b $  and in the future light-cone of  $ a $. The so-called correlation function is then given by
	\begin{equation}\frac{\delta}{i\delta J(x_k)}\ldots\frac{\delta}{i\delta J(x_j)}\langle out|in\rangle=\langle out|T((\hat\varphi,\ldots,\hat\varphi)(x_k),\ldots,(x_j))|in\rangle\label{corr1}\end{equation}
	where  $ T $  is the chronological ordering operator defined by \begin{equation}T((\hat\varphi_1 \hat\varphi_2\ldots\hat\varphi_n )(x_1,x_2,\ldots x_n))=\hat\varphi_1(x_1)\hat\varphi_1(x_2)\ldots \hat\varphi(x_n),\end{equation}
	with  $ x_k $  lying the future light-cone of  $ x_{k+1} $  and  $ x_{k} $  and  $ x_{k+1} $  lying both in the future light-cone of  $ x_{k+2} $  and so on.
	Furthermore, \begin{equation}\langle out|in\rangle=N\int_{\mathcal{P}_{ab}}\mathcal{D}_{\Theta,Z}\varphi \;\Theta(\varphi,J).\label{pathintegral}\end{equation}

\end{definition}

Using Eqs.  (\ref{commutatorpeierls}) and (\ref{perturb}), one finds that to leading order  $ \mu(\varphi)=\frac{1}{\sqrt{|sdet G^+(\varphi)|}} $, where  $ sdet $  is the super-determinant of the advanced Green's function  $ G^+ $, which is the inverse of  $ S'' $  in the expansion, see \cite{DeWittbook}.	

\begin{remark}At first, the  $ |in\rangle $  and  $ \langle out|  $  states seem a bit undefined. However,  one may employ a method that is usually used in quantum gravity, see e.g. \cite{Leutwyler,DeWitt, Hartle,Feng,Hartle}: It is usually possible to introduce  canonical phase space variables  $ p,q $  for a classical field theory with a given action. Using phase space variables, we can write the action functional as 
	\begin{equation}S(p,q,t_b)=\int_{t_a}^{t_b}dt\left( p_i \dot{q}^i-H(p,q,t)\right)\end{equation}
	
	where  $ H $  is the Hamiltonian. At least for gravity, a suitable functional derivative with respect to the fields then shows that the  $ \langle out| $  states fulfill a functional Schr\"odinger equation for a Hamilton operator  $ \hat H $, see \cite{Hartle,Feng}.

	Solutions of the Schr\"odinger  equation  fulfill the superposition principle. If one can define a scalar product for  $ \langle out| $, one may be able to expand the solutions of this functional Schr\"odinger equation in an orthonormal series. One may possibly use this to show that  $ \langle out |$  is an element of a complete metric space with a scalar product, i.e. a Hilbert space.
	
	With  $ \langle out| $  and   $ |in\rangle $  in a Hilbert space, Eqs.  (\ref{expect}) and (\ref{corr1}) then become basically definitions for the field operator  $ \hat\varphi $. Showing that  $ \hat\varphi $  satisfies the Wightman axioms for any possible action would still require some work. However, in case where one can compute correlation functions from theories satisfying the Wightman axiom, one may show equality by simply computing the correlation functions from Eqs.  (\ref{expect}) and (\ref{corr1})  and then compare. 
\end{remark}
\begin{remark}The definition of quantum field theory above was given for bosonic fields. Path integrals over fermionic fields can be described with the help of functional integrals over so-called Grassmann variables. The path integral of DeWitt-Morette and Cartier was extended for them in \cite{functionalintegration}.
\end{remark}Very often, a path integral like the one from theorem (\ref{integrator}) does not exist, i.e. will be divergent. Assume, you want to compute an integral  $ \int_{t_a}^{t_b}f(x)dx $  of a function  $ f(x),x,f\in\mathbb{R} $  but this integral diverges. One could then make a Taylor expansion of  $ f(x) $  around some  $ x_0 $. But in the case of quantum field theory, even the functional integral over the perturbation series in Eq.  (\ref{perturb}), with the expansion of the action around a background is often divergent. As a solution, physicists have come up with the definition of renormalization. Its mathematical idea is basically as follows \cite{functionalintegration}:

\begin{definition}Let   $ S(m,e,\varphi) $  be the action of a classical system, where  $ m,e $  may be constants with some physical dimension, like mass  $ m $, charge  $ e $  and maybe others, like Einstein's gravitational constant  $ \kappa $. Let  $ C(m,e,\varphi )$ be a functional which is added to the action  $ S $, i.e. define  $ S_{total}(m,e,\varphi)=S(m,e,\varphi)+C(m,e,\varphi) $  and expand  $ S_{total}(m,e,\varphi+\varphi_0) $  according to (\ref{perturb}) around a classical solution  $ \varphi_0 $  of  $ S(m,e,\varphi+\varphi_0) $. Denote the terms in the summation of the expansion of   $ S_{total}(m,e,\varphi+\varphi_0) $  as  $ S_{t1},S_{t2},\ldots, S_{tn} $. 
	
	$ C(m,e,\varphi $ ) is called counter-term if the functional integrals of Eq.  (\ref{pathintegral}) converge for all the functions  $ S_{t1},S_{t1},S_{t2}, \ldots S_{tn} $  if they are used as actions in  $ \Theta(\phi,J) $. For each addend in the series, one may need to use different constants  $ m,e $  or others in  $ C(m,e,\varphi) $  to achieve convergence of the functional integrals. A theory is called renormalizable, if the functional integral for the entire expansion of $ S_{total}(m,e,\varphi+\varphi_0) $ becomes finite when the values of the constants $m,e,\ldots$ do not change at different orders $ S_{t1},S_{t1},S_{t2}$ of the expansion. \label{counterterm}
\end{definition}

\begin{remark}For the physicist, this implies that he can measure the constants  $ m,e,\kappa $  and so on by making an experiment that is supposed to be the result of a functional integral over e.g.  $ S_{t1} $, i.e. at low orders of the perturbation expansion. Then the renormalized quantum field theory with the counter-terms will make predictions for the results of the functional integral at the higher orders with the functional integration of  $ S_{t2} $  and so on, since the same counter-term can be used to compute finite correlation functions for the higher orders in perturbation theory. 
	
	The reader with knowledge on quantum field theory books may have found the above definition insufficient, however, there are many methods to renormalize the path integral and find the correct counter-term. Since we are working with complex integrands, most of these methods are based on analytic continuation. For more information about the various methods, the reader is advised to consult a quantum field theory book \cite{DeWittbook}.
\end{remark}

\begin{remark}Sometimes, an action  $ S(\varphi) $  is invariant on a certain gauge transformation acted out on the fields, e.g. an  $ U_1 $  transformation, or a transformation with the diffeomorphism group. In that case, the measure must be adjusted in order to avoid over-counting. DeWitt \cite{DeWitt22} and later Fadeev and Popov \cite{Faddeev} have  discovered a notion how to write functional integrals over such spaces with gauge fixing terms and so-called ghost fields.
\end{remark}
\begin{remark}If one applies the usual techniques for gauge fields on the functional integral for gravity, one of these functional integral goes formally over the diffeomorphism group. The latter is not isomorphic to a Banach space in general. Therefore, this procedure does not make sense in the path integral framework of DeWitt-Morette. Fortunately, in \cite{DeWitt50years}, DeWitt has found a technique to make his results compatible with the formalism of DeWitt-Morette and Cartier. The end result is the same: For an action with a bosonic/(fermionic) field which is invariant under a gauge transformation, one has to add a gauge fixing term and another functional with ghost terms.
	
\end{remark}

\begin{remark}Often, we are evaluating  path integrals on Euclidean space
	\begin{equation}
		Z_{e}=\int_{\mathcal{P}_{ab}} d\phi\exp{(-S_{eu})}
	\end{equation}where the metric involved in the action functional  $ S_{eu} $  is now Riemannian. The procedure that sometimes connects Lorentzian and equivalent Euclidean path integrals is called Wick-rotation. However, it is often not permitted to go from a Lorentzian path integral to an Euclidean one, evaluate that and then go back to Lorentzian space. In axiomatic quantum field theory, the Osterwalder-Schrader theorem \cite{Osterwalder} gives some conditions when there is an isomorphism of quantum field theories on Minkowski and on Euclidean space-time. DeWitt could only prove this isomorphism to one loop order for non-gauge theories with the path integral formalism, see his contribution in \cite{functionalintegration}. Recently, Visser has shown in \cite{Vis2,Vis3} that for curved space-times the there is often no isomorphism between Euclidean and Lorentzian metrics and one can not transform a path integral involving an arbitrary Euclidean metric back to a functional integral with an that contains fields in Lorentzian space-times. Visser proposed to fix a time-like vector field normed to unity and use theorem \ref{theoremvector} to get an Euclidean metric from a Lorentzian one and back.
\end{remark}
\subsection{The path integrals that we investigate in this article}

In this article, we are investigating the path integral of quantum gravity.
\begin{definition}	Let the point sets  $ a $  and  $ b $  be coordinates of initial and final points where we can define 3 surfaces  $ S_1 $  and  $ S_2 $  with three metrics  $ \gamma_{ij}^a $, and  $ \gamma_{ij}^b $  and the path integral goes over all 4 metrics  \begin{equation} g_{\mu\nu}=\left(\begin{matrix}
			-\alpha+\beta_k\beta^k	&  \beta_j\\
			\beta_i	& \gamma_{ij}
		\end{matrix}\right) \label{dewittparam} \end{equation}   of cobordisms  $ M $  between  $ S_1 $  at  $ a $  and  $ S_2 $  at  $ b $, see \cite{DeWittbook1, DeWitt, DeWitt22, DeWitt3,Geroch,Leutwyler, Hartle, Feng}:
	
	\begin{equation}
		\langle out|in\rangle=Z(\gamma_{ij}^b(\overline b,\overline b'),t_b)=\int_{\mathcal{P}_{ab}^\gamma}\mu(g_{\mu\nu}) dg_{\mu\nu}\exp{(iS_g(g_{\mu\nu}))}\label{eq:gavipath}
	\end{equation}
	$S_g$ is Einstein's action, $ \mathcal{P}^\gamma_{ab} $  is now the space of  $ 4 $  metrics  $ g_{\mu\nu} $  that define a  chronological weak and causally compact cobordism with  $ \gamma_{ij}^a $  at  $ a $  and  $ \gamma_{ij}^b $  at  $b$, and $\overline {b},\overline{b}'\in b $.  
\end{definition}

If we try to investigate path integrals over space-times with singularities (that are no cobordisms) later in the article, we will mention it explicitly.
\begin{remark}$ S_g $ yields classical equations of motion that are diffeomorphism invariant. In order to evaluate the path integral, it is understood that one has to add gauge-fixing and ghost terms to the gravity action and define additional functional integrals over the ghost fields along the lines of \cite{DeWitt50years,functionalintegration}.
\end{remark}
\begin{remark}	In order to define Eq.  (\ref{eq:gavipath}) properly, we should make sure that we have a metric for the space of 4 metrics, because only then we can get a Banach space for the paths that is needed in the functional integration. Additionally, a Banach space would help defining functional derivatives with functionals of metrics.
\end{remark}

The analysis of spaces of metrics in quantum gravity began with the early article from DeWitt \cite{DeWitt} and Fischer  \cite{fischer} where the space  $ Riem(M^3)/Diff(M^3) $  of Riemannian  $ 3\times3 $  metrics  $ \gamma_{ij}^b $  divided by the diffeomorphism group was partially analyzed.

Since quantum gravity should be invariant under coordinate transformations and diffeomorphisms, metrics related by the action of a coordinate transformation or a diffeomorphism define the same geometry. The quotient space  $ Riem(M^3)/Diff(M^3) $  turned out to be much more difficult than the space  $ Riem(M^3) $. Fischer found that quotient singularities arise at the boundary of   $ Riem(M^3)/Diff(M^3) $  since metrics that define highly symmetric spaces can be transformed easily into different metrics that yield the same space. Hence in contrast to  $ Riem(M^3) $, the space   $ Riem(M^3)/Diff(M^3) $   was not a manifold. In the same conference proceedings where the article of Fischer can be found, DeWitt \cite{DeWittsuperspace} argued  $ Riem(M^3)/Diff(M^3) $  can be extended such that it becomes a manifold (see also \cite{DeWitt11} for an introduction).

More recent articles on that topic are the ones from Giulini \cite{Giulini} and Anderson \cite{Andersonconfig}. From the mathematical side, one has the articles of Ebin \cite{ebin} that came shortly after DeWitt's first investigation, and Clarke \cite{Clarke} , and Bauer, Harms and Michor\cite{michor}  

\begin{remark}The articles of and DeWitt that define a metric for the space of metric tensors often write about the space of  $ 3\times 3 $  metrics. However, in the path measure of Cartier and DeWitt-Morette, one has  $ 4\times4 $  metrics, or, since the dimensionality of the metrics should not be a restriction to define a path integral over metrics, we need a metric for the space of Lorentzian  $ n\times n $  dimensional Lorentzian metrics. Fortunately, the details of the works of Fischer, Ebin, and Bauer make clear that their considerations hold for spaces of metrics of manifolds with arbitrary dimension. The  authors state that their metrics for  $ Riem(M^n)/Diff(M^n) $  are often only defined for closed manifolds. The Lorentzian cobordisms of gravity are, however, just causally compact. If one looks at \cite{ebin,Clarke,michor,fischer,DeWittsuperspace} one notes that their considerations hold for compact manifolds. Furthermore, especially the metric of  $ Riem(M^n)/Diff(M^n) $  in \cite{DeWittsuperspace} is also finite for asymptotically flat space-times where one has appropriate fall-off conditions for the metrics  $ g $. We will simply add such metrics  $ g $  to the space of paths, where the metric of  $ Riem(M^n)/Diff(M^n) $  is finite.
\end{remark}The  $ 4\times4 $  metrics of the cobordisms in quantum gravity are Lorentzian. As we have noted, the Wick rotation is in general not permissible from an arbitrary Euclidean manifold to a Lorentzian one. We can, however, adopt the procedure outlined by Visser \cite{Vis3}.

\begin{definition} We use theorem \ref{theoremvector} to augment a Lorentzian metric  $ g_L $  with an arbitrary unity normed time-like vector field $v$ and call the space of these vector fields as $V_{gL}$: 
	
	\begin{equation}g_\epsilon^{ab}=(g_L)^{ab} -i\epsilon v^{a}v_{b}\end{equation}
	A Riemannian metric is then described by the limit  $ \epsilon=2i $  and the Lorentzian case is recovered by the limit  $ \epsilon=0 $. That way, we can convert any Lorentzian cobordisms to an Euclidean metric and then apply the results for the spaces of metrics. 
\end{definition}Since the action of gravity is diffeomorphism invariant, two metrics that are related by a diffeomorphism should have the same norm.

\begin{definition}We define the metric for  $ g_{L1}^{\mu\nu},g^{\mu\nu}_{L2} $  in   $ L(M^4)/Diff(M^4) $, the space of differentiable Lorentzian Metrics for  $ M $  divided by the Diffeomorphism group as
	\begin{equation} d_{L}(g_{L1}^{\mu\nu},g^{\mu\nu}_{L2})=inf_{v\in V_{g_{L1}},v'\in V_{g_{L2}}}(d_{Riem/diff}(g_{1\epsilon=2i}^{\mu\nu}, g_{2\epsilon=2i}^{\mu\nu})).\label{riemdiff}\end{equation}
\end{definition}We now only have to find  $ d_{Riem/diff} $. There are several choices. One may look at the set of Sobolev Riemannian metrics that can be defined for the geometry on  $ M $.
\begin{definition}Sobolev Riemannian metrics are Riemannian metrics whose derivatives are square integrable up to order  $ s, s>n/2 $  for an  $ n\times n $  dimensional metric over the volume of  $ M $. For these metrics, one can get a metric  $ d^s_{Riem} $, see \cite{ebin}, \cite{fischer}. 
\end{definition}According to \cite{fischer} from  $ d^s_{Riem} $  one can then deduce a metric  $ d_{Riem} $  that induces the topology of  $ Riem(M) $  by:
\begin{definition}
	\begin{equation}
		d_{Riem}(g_{1\epsilon=2i}^{\mu\nu},g_{2\epsilon=2i}^{\mu\nu})=\sum_{s>n/2}\frac{1}{2^s}\frac{d^s_{Riem}(g_{1\epsilon=2i}^{\mu\nu}, g_{2\epsilon=2i}^{\mu\nu})}{1+d^s_{Riem}(g_{1\epsilon=2i}^{\mu\nu}, g_{2\epsilon=2i}^{\mu\nu})}\label{Fischermetric1}
	\end{equation}
	From the metric for  $ Riem(M) $, one can then get a metric  $ d_{Riem/diff} $  by 
	\begin{equation}d_{Riem/diff}(g_{1\epsilon=2i}^{\mu\nu},g_{2\epsilon=2i}^{\mu\nu})= \inf_{\zeta\in Diff_F(M)}\left(d_{Riem}\left(g_{1\epsilon=2i}^{\mu\nu},\zeta*g_{2\epsilon=2i}^{\mu\nu}\right)\right)\label{driem}.\end{equation}
	
\end{definition}
\begin{remark}Note that this works only if  $ g_{1\epsilon=2i}^{\mu\nu},g_{2\epsilon=2i} $  have square integrable derivatives up to order  $ s>n/2 $.
\end{remark}

In \cite{DeWittsuperspace}, DeWitt has found another choice as a metric for  $ Riem(M) $  which yields the geometries for   $ Riem(M^n)/Diff(M^n) $, in the sense that the geometries are identical if and only if the metric between them is zero. It has the advantage that one does not need square integrable derivatives of the metrics.

\begin{definition}For  $ g_{\mu\nu} $  and  $ g_{\mu\nu}+\delta g_{\mu\nu}\in Riem(M) $, one has
	\begin{equation}
		d_{Riem}(g_{\mu\nu},g_{\mu\nu}+\delta g_{\mu\nu})
		=\int \int d^4x d^4x' G^{\mu\nu\alpha\beta}(g, x,x')\delta g_{\mu\nu}(x)\delta g_{\alpha \beta}(x') \label{dewittmnetric}
	\end{equation} with
	\begin{equation}
		G^{\mu\nu\alpha\beta}(g,x,x')=\frac{1}{2}\sqrt{g(x)}(g^{\mu\alpha}(x)g^{\nu\beta}(x)+g^{\mu\beta}(x)g^{\nu\alpha}(x)+\lambda g^{\mu\nu}(x)g^{\alpha\beta}(x))\delta(x,x'),
	\end{equation}
	where  $ \lambda\neq-2/n $. 
	
\end{definition}
\begin{remark} From this metric, the measure of Euclidean quantum gravity can be constructed see \cite {Hamber}.
\end{remark} 

\begin{remark}The components of metrics tensors are of course coordinate dependent.  Therefore, one also has to use Eq.  (\ref{driem}) to get the final metric for  $ Riem(M^n)/Diff(M^n) $. Practically, this means that if one wants to determine the distance between two metrics  $ g $  and  $ g'=g_{\mu\nu}+\delta g_{\mu\nu} $  one has to fix one metric  $ g $, and then find the coordinate transformation  $ \zeta $  such that  $ d(g,g') $  is at a minimum.
\end{remark}

Similarly as for fields  $ \varphi(x)\rightarrow\mathbb{C} $, one can now select different times  $ t_1,t_2\in\mathbb{R},t_1<t_2 $  and fix Cauchy data for the components of the metric  $ g_{\mu\nu}(x),x\in M $  for points  $ a=(t_a,\mathbf{x_a})\in M $  and  $ b=(t_b,\mathbf{x_b})\in M $. 

\begin{remark}It should be noted here that the first attempt to extend the path integral formalism of DeWitt-Morette for fields was in fact an extension to infinite dimensional spaces dual to nuclear spaces in order to try to write the gravitational amplitude \cite{ClarkeGrav}.
\end{remark}

Again, since the action is gauge invariant, one has to add gauge fixing terms and Grassmannian ghost fields to the integrand, according to  \cite{DeWitt50years,functionalintegration}. Additionally, one has to supply an appropriate counter-term to cancel one-loop divergences, see 't Hooft\cite{thoofvelt}.

\begin{remark}If one then does a perturbation expansion around some background metric  $ g0 $, one finds that the theory can only be renormalized up to second order\cite{thoofvelt}. One should note that inconsistencies arising in a perturbative series do not imply that the path integral itself is inconsistent. The path integral could be non-perturbatively consistent within a framework of asymptotic safety \cite{Weinberg}. 
\end{remark}
\begin{remark}One can express the action in terms of a Hamiltonian. In eq. (\ref{dewittparam}), one may set  $ \alpha=1,\beta_i=0 $, which corresponds to a deliberate choice of a certain coordinate system. Then, one can express the action of the functional integral in terms of  $ 3 $  metrics  $ \gamma_{ij} $  only. One can then rewrite the action such that it includes the ADM Hamiltonian of gravity \cite{ADM} which  depends solely on  $ \gamma_{ij} $. By functional differentiation of the path integral of gravity with the Hamiltonian in the integrand with respect to the metric, one observes that the amplitude solves the Wheeler-DeWitt equation 
	\begin{equation}i\partial_t Z(\gamma(b))=\hat H_g Z(\gamma(b)),\label{wdw2334}\end{equation}  which is a functional Schr\"odinger equation, with  $ \hat H_g $  as Hamiltonian operator of gravity \cite{Hartle,Feng}. In Eq.  (\ref{wdw2334}),  $ b $  are are points in  $ S_2 $  and  $ \gamma $  is a short notation for the metric tensor field  $ \gamma_{ij}(b) $.
\end{remark}
\begin{remark}Note that  $ \hat H_g $  is usually numerically zero if no Gibbons Hawking York boundary term, see\cite{York}, was added to the action from which  $ \hat H_g $  was derived. Without such boundary term, the Wheeler-DeWitt equation  then describes no time evolution for its operators and states\cite{DeWitt,Feng}.

	One could  chose a suitable canonical transformation, which renders the Hamiltonian non-zero, but by the Groenewold-van Hove theorem, each of those canonical transformations would result in a different quantum theory,see \cite{Groenewold, Haj}. 
	
	So one may infer from this ambiguity of the gravitational path integral that a boundary should be present in the space-time.
	
	Such a boundary must not be at infinity, but can, for example, also define a hole in the space-time that emerges from quantum mechanical processes. 
	
	Moreover, in Eq.  (\ref{eq:gavipath}) the  $ 4\times4 $  metrics over which is integrated are actually defined! as Lorentzian cobordisms  $ M $  which are manifolds with boundaries at the sets of the start points and endpoints of the time evolution. So by the very definition of the path integral of quantum gravity, boundaries should be present in the space-time even if one considers processes which start at some time $t_1$ and with a given three metric $\gamma_{ij}(a)$ (i.e. the space of a universe at early time) and end at some time $t_2$ with a given three metric $\gamma_{ij}(b)$ (i.e. the space of a universe at late time) for a closed universe. This was also correctly noted by Hartle and Hawking in \cite{Hartle} who did, however, use a Hamiltonian without a boundary term even if their action such a boundary term supplied. That a boundary term in the action also leads to a non-zero Hamiltonian was mentioned at first in \cite{Feng}. If one were to put e.g. the late configuration  $\gamma_{ij}(b)$ at $t_2$ to infinity, then one also would have to supply a boundary term, as was noted by DeWitt \cite{DeWitt}.
\end{remark}

\begin{remark}Sometimes we will use the path integral of Euclidean quantum gravity which goes over Riemannian metrics. In that case, however, one should restrict the path integral to ones that can be rotated back into Euclidean space-time, which is not possibly for every Riemannian metric\cite{Vis2,Vis3}.
\end{remark}
\begin{remark}	One should also note that Gibbons Hawking and Perry discovered that, without adding Ghost and Gauge fixing terms, the Euclidean path integral of quantum gravity has a diverging integrand if the action is not evaluated within a perturbation series \cite{Zetafunction}. However, recently in \cite{finiteeu,finiteeu2} ghost and gauge fixing terms were added to the action and the calculation of the amplitude was repeated. The author of  \cite{finiteeu,finiteeu2} argues that the path integral of Euclidean quantum gravity  converges if ghost and gauge fixing terms are added. If this result is correct, it would be interesting, since it would indicate that a non-perturbative evaluation of the path integral of Euclidean quantum gravity is possible. Amplitudes of Euclidean field theories are generally useful to compute the time independent ground state of quantum systems, where one should not have much difficulties of transforming the system back to Lorentzian metrics. Euclidean amplitudes are also
	are related to the canonical partition function that determines the statistical mechanics of a quantum field theory, which can be used for thermodynamical considerations. 
\end{remark}

Since many claims about topology changes come from the String theory literature, we also take a look at the path integral for the nonlinear sigma model:
\begin{definition}
	The Polyakov action for a closed string on a target space-time background  $ M $  with metric  $ G_{\mu\nu} $  is \begin{equation}S_p=\frac{1}{T}\int_\Sigma d^{2}\Sigma\sqrt{-\tilde\gamma}\tilde\gamma_{uw}\partial^{u}X^{\mu}(\sigma,t)G_{\mu\nu}\partial^{w}X^{\nu}(\sigma,t),\label{poly}\end{equation} where  $ T $  is a constant called string tension and  $ \tilde\gamma_{uw} $  is a  $ 2\times2 $  metric tensor of a world-sheet. The latter is parameterized by two coordinates  $ \sigma\in [0,2\pi],t\in\mathbb{R}_{\geq 0} $. The functions  $ X^\mu(\sigma,t):([0,2\pi],\mathbb{R}_{\geq 0})\rightarrow M $  describe a mapping of the world-sheet $ \Sigma $   into the target space  $ M $, see \cite{deligne}.
	
	The path integral over the Polyakov action 
	\begin{equation}
		\int_{\mathcal{P}_{ab}^{X^\mu}} \int_{\mathcal{P}^{\tilde{\gamma}_{uw}}_{ab}} d\tilde\gamma_{uw}dX^\mu \exp{(iS_p)}
	\end{equation}
	
	thus involves a functional integration over  $ X^\mu $  and  $ \tilde\gamma_{uw} $.  The integration over the world sheet metric is similarly defined as the integral in quantum gravity. The space of paths  $ \mathcal{P}^{\tilde{\gamma}_{uw}}_{ab} $    is the space of 2 metrics of cobordisms with one dimensional surfaces as boundaries at times  $ t_a $  and  $ t_b $,  $ t_a<t_b $. If one defines two 2 dimensional point sets  $ a=(\sigma_a,t_a) $  and  $ b=(\sigma_b,t_b) $  where the points  $ a $  lie in the past of  $ b $, the space of paths  $ \mathcal{P}^{X^\mu}_{ab} $  is the space of fields  $ X^\mu(x) $  with  $ x=(\sigma,t) $  as a two dimensional point in the future light-cone of  $ a $  and in the past light-cone of  $ b $. 
\end{definition}
\begin{remark}In a scattering process, the world-sheet appears to be a manifold with boundaries. Thus, the Polyakov action probably also should be augmented by a boundary term for a one dimensional boundary at the two ends of the world-sheet  $ t_0 $  and  $ t_1 $, even though this is often not done in the string theory literature. The procedure to convert the classical theory given by the Polyakov action into an equivalent one based on a canonical Hamiltonian and write a functional Schr\"odinger equation by functional differentiation of the amplitude, is  possible, of course also with the Polyakov action, and not only with gravity. Without a boundary term in the world-sheet, the functional Hamiltonian, at least the part for the world-sheet also vanishes.
	String theorists usually argue, albeit on different grounds, that the ends of the string would have to be fixed at D-branes (which can form boundaries). 
\end{remark}
\begin{remark}The path integral over the  $ X^\mu(x) $  in the amplitude for the Polyakov action is different from the usual path integral in field theory, since the field has the target space  $ M $  as its co-domain, which is usually a curved Lorentzian manifold with complicated topology and not just  $ \mathbb{C} $. The domain of  $ X^\mu(x) $  is the 2 dimensional world sheet  $ \Sigma $, which can have a complicated topology, similar as in ordinary 4 dimensional quantum field theory. 
\end{remark}
\begin{remark}The Polyakov action is invariant under world-sheet diffeomorphisms and Weyl transformations. Accordingly, one has to add ghost actions to the path integral that follow rigorously from DeWitt's procedure \cite{DeWitt50years}, in order to define it correctly.
\end{remark}Usually, string theorists write their path integrals over Euclidean world sheets and space-time backgrounds, i.e. they do a Wick rotation before the path integration. It turns out that the amplitude then becomes a sum over topologically distinct Euclidean world-sheets.

\section{Path integrals over metrics and singularities}\label{s1}
The earliest argument against topology changes was made by DeWitt and Anderson in \cite{Topology} for a certain case with a conical singularity.

In \cite{DeWitt9} and \cite{DeWittbook}, DeWitt tried to come forward with a new argument based on homotopy considerations. But this argument only came in form of two short comments, from which it is not really clear to which topology changes and to which quantum field theories this argument really applies. Here we will analyze it in detail and apply these considerations on quantum gravity, quantum field theory in  $ 4 $  dimensions and on string theory. This requires from the reader to be familiar with homotopy theory.

A result from Schulman  \cite{Schulman} shows that 
\begin{theorem}For path integrals over  paths  $ q(a,b)\rightarrow M $  from  $ a $  to  $ b $  that map to a non-simply connected configuration space  $ M $, the amplitude becomes a superposition  
	\begin{equation}
		Z=\sum_\alpha  D(\alpha)K^\alpha\label{eq:DeWitmorette}
	\end{equation}
	
	of partial amplitudes  $ K^\alpha $, where in each  $ K^\alpha $  only paths of one homotopy class  $ \alpha $  are considered. Each of these homotopy classes consists of paths that are homotopic to each other.
\end{theorem}
\begin{proof}See \cite{Schulman}\end{proof}
Furthermore, one has
\begin{theorem}The phase factors  $ D(\alpha) $  must be scalar unitary representations of the fundamental group of the configuration space  $ M $.
\end{theorem}
\begin{proof}See the proof of Laidlaw and DeWitt-Morette in \cite{Laidlaw}. provided that the path integral is finite, their result holds for path integrals with arbitrary actions and configuration spaces which are multiply connected, arc-wise connected, locally arc-wise connected, and locally simply connected and metrizable.
\end{proof}

We now consider DeWitt's arguments from \cite{DeWittbook,DeWitt9}.
Certainly, homotopy equivalence does not imply diffeomorphism equivalence. For example, exotic spheres are smooth closed manifolds that are homotopy equivalent to  $ S^n $  but they are not diffeomorphic to it. On the other hand, if two spaces are not homotopy equivalent, they are also not diffeomorphic.

Therefore, it appears at first as reasonable to use an argument for path integrals in non-simply connected configuration spaces to discuss topology changes.

However, unlike in quantum mechanics, where one has paths  $ q(a,b)\rightarrow M $, usual quantum field theories have paths  $ \varphi:M\rightarrow \mathbb{C}^n $, where  $ M $  is the space-time. If the space-time changes its topology, the configuration space  $ \mathbb{C}^n $  does not undergo a topology change. 

Usually, string amplitudes are computed with world-sheets  $ \Sigma $  that change their topology in time and one has paths   $ X^\mu:\Sigma \rightarrow T^n $, where  $ \Sigma $  is the 2 dimensional world sheet and  $ T $  is the target space.

That the string theory amplitude can be written over topology changing world-sheets  $ \Sigma $  gives a counterexample to the assumption that the formula of Laidlaw and DeWitt Morette would forbid topology changes in  $ \Sigma $  for fields  $ X^\mu:\Sigma\rightarrow \mathbb{C}^n $.

A slightly different situation arises for string theory amplitudes if the target space  $ T $  is a Lorentzian manifold which is a union of a space-time manifold  and possibly singularities that arise because  $ T $  has no closed time-like curves and interpolates between space-like hyper-surfaces that are not diffeomorphic and thus not homotopic to each other. 

If we leave the problems with the singularities aside and assume that paths of all matter fields simply circle around them, the paths $ X^\mu:\Sigma\rightarrow T $ describe exactly the same situation as in the problem considered by Laidlaw and DeWitt-Morette, where the particles are at some time confronted with the 'hole' of the solenoid from the Aharanov-Bohm effect.

One can denote the paths $ X^\mu:(\sigma,t)\rightarrow T $ of string theory by $ q(a,b) $  with $ a=(0,0) $ and $ b=(2\pi,t_b) $. Then the proof of Laidlaw and DeWitt-Morette works basically unchanged. The only difference between the string and particle case can be that in case of closed strings, a target space singularity may be encircled in two different ways by the string. Similarly as in the particle case, the singularities can be encircled one or several times by paths with the curve parameter $ \tau $ going from $ \tau=0 $ to  $ \tau=\infty $. In the string theory case, the singularity can also be encircled for one or several times if one moves the curve parameters $ \sigma $  from  $ \sigma=0 $  to  $ \sigma=\pi $. This does, however, not affect the proof of Laidlaw and DeWitt-Morette at all since these two paths can not be deformed into each other and thus just belong to different homotopy classes. Furthermore, the proof does not depend on the number of dimensions of the quantum field theories considered.

As DeWitt notes in \cite{DeWittbook}, there are no possibilities of interpolating paths between different homotopy classes of paths. But that is exactly the situation for which the amplitude of  Laidlaw and DeWitt-Morette is constructed. 

Interestingly, DeWitt has noted himself in \cite{DeWittbook} that topology changes may be followed by the formula of DeWitt and Laidlaw if space-time would be embedded in a higher dimensional manifold. Indeed, topology changing $ n-1 $ surfaces can be embedded in the $ n $  dimensional space-time (which must by Geroch's theorem, be a union of a manifold and singularities if it is Lorentzian and has no closed time-like curves). The phase factors in the formula of Schulman, Laidlaw and DeWitt-Morette are then given by a function of the fundamental group of the entire  $ n $  dimensional target space-time. This target space-time then has a fixed topology, from which one can compute its fundamental group.

However, in string theory and quantum gravity, one has not just the case of paths of embedding functions $X^\mu$ mapping to the target space, but one also has paths of world-sheet metrics $\tilde\gamma_{uw}$.

If one has to integrate a path integral over world-sheet metrics with singularities, the situation becomes more difficult. We will now prove the following
\begin{theorem}If there does not exist a metric $ d $ for metric tensors that have  singularities and whose restriction to points where they are regular describes spaces in $ Riem(M)/Diff(M) $, then one can not define a functional integral over such metric tensors with the formalism of Schulman, Laidlaw and Cartier and DeWitt-Morette\label{thmm1}.
\end{theorem}

\begin{proof}Let us consider  an $ n $  dimensional Lorentzian space-time $M$ that interpolates between  $ n-1 $ dimensional hyper-surfaces of different topology and have no closed time-like curves. By Geroch's theorem, the metric must have singularities at some points. We denote the space of these singular interpolating metrics by  $ \tilde{\Theta}(M) $. This is the configuration space of paths.
	
	In order to apply the formula of Laidlaw and DeWitt-Morette to a functional integral over a metric, we  would need to be able to define the fundamental group of  $ \tilde{\Theta} $. For this we need to define a base-point  $ x_0\in\tilde{\Theta}(M) $  and consider various loops that start and end at  $ x_0 $. Furthermore, we would need  $ \tilde{\Theta}(M) $  to be multiply connected, arc-wise connected, locally arc-wise connected, and locally simply connected and, especially, metrizable.

	In order to define the fundamental group, we need a continuous function \begin{equation}I:[0,1]\rightarrow \tilde{\Theta}\end{equation} 
	and assume that  $ I(0)=I(1) $.
	
	For the definition of  $ I $  as a continuous function,  we need a metric  $ d $  for every  $ g\in\tilde{\Theta}(M) $  for which  $ \exists\tau $  such that  $  I(\tau)=g $. And this metric  $ d $  would have to be such that \begin{equation}\lim_{\epsilon\rightarrow0} d(I(\tau),I(\tau+\epsilon))\rightarrow0.\end{equation} 
	
	The procedure to define a time-like vector field  $ V $  and make Wick rotation of the metric of an interpolating space-time  $ g $  in  $ \tilde\Theta $  according to Visser's prescription would then become ill defined, since at the metric singularity, one can not define what a time-like vector is. Furthermore, the Euclidean Metric in Eqs.  (\ref{driem}) for Sobolev Riemannian metrics becomes ill defined if it is used for  metrics with singularities, since their derivatives are not square integrable. One could try to use the DeWitt metric from Eq.  (\ref{dewittmnetric}). As DeWitt writes, it is the only possible metric for metric tensors that does not involve derivatives and which is zero if the geometry of two metrics coincides. However, the DeWitt metric often becomes ill defined for singular metric tensors.

	One could be tempted to remove  singularities of  $ g $  by hand and just compute  $ d $  for the remaining differentiable manifold, but then the metric  $ d $  could no longer distinguish between a space having a singularity at a given point  $ p $  and another regular space that has a hole at  $ p $. Thus the metric would not longer induce the topology. Two spaces could have different geometry and yield a distance of zero.
	
	One may argue that we do not need the formula of Laidlaw and DeWitt-Morette, since if the manifolds  $ S_1 $  and  $ S_2 $  have distinct topology and if we exclude interpolating space-times with closed time-like curves from the summation, all metrics over which we integrate the path integral have singularities, so they may be in the same homotopy group of the configuration space of metrics.
	
	However, we also need a well defined norm for the space of metrics that we need to transform with an affine transformation to a Banach space before we can define the path integral with the formalism of DeWitt-Morette and Cartier, see theorem \ref{integrator}. If they all have a divergent or otherwise not well defined norm, this does no longer work. 
	
\end{proof}This result is interesting because it does not involve the form of the action in any way. Instead, it solely depends on other, more general measure theoretic details of the path integral formalism, see \cite{functionalintegration} and references therein. 
Combined with Geroch's theorem, it means that, provided one can not find a metric for Lorentzian metric tensors which include singular entries, one can not integrate the path integral over Lorentzian metrics that interpolate between topologically distinct space-like hyper-surfaces without closed time-like curves. 
\begin{remark}
	A reader who has read DeWitt's work \cite{DeWitt} on the singularity of the Friedmann cosmos in the framework of the Wheeler-DeWitt equation may wonder how his results fit into the statements given here. 
	
	In his article, DeWitt was able to rewrite the Wheeler DeWitt equation in a certain coordinate system as 
	\begin{equation}
		\left(-\frac{\delta^2}{\delta\zeta^2}+\frac{32}{3\zeta^2}\overline G^{AB}\frac{\delta}{\delta\zeta^A}\frac{\delta}{\delta\zeta^B}+\frac{3}{32}\zeta^2  \;^{(3)}R\right)\psi(\gamma),\label{wdw123}\end{equation}
	where \begin{equation}\zeta=\sqrt{\frac{32}{3}}\gamma^\frac{1}{4}\end{equation} 
	is a time-like coordinate and  $ \zeta^A $  are five coordinates orthogonal to  $ \zeta $, and \begin{equation}\overline{G}_{AB}=tr(\gamma_{ij}^{-1}\gamma_{ij},_A\gamma_{ij}^{-1}\gamma_{ij},_B)\end{equation}
	The comma denotes partial derivatives with respect to  $ x^A $, the notion  $ \frac{\delta}{\delta\zeta^A} $  describes a functional derivative, i.e. \begin{equation}\frac{\delta}{\delta\zeta^A}\psi(\gamma)=\int \frac{\delta}{\delta\zeta^A}\psi(\gamma_A(x))\phi_{A}(x)d^3x\end{equation}
	with $ \phi_A(x) $  as an arbitrary function of  $ x $. Note that DeWitt omitted the boundary terms of the cobordism and the specification of the space-like hyper-surfaces $ S_1 $ and  $ S_2 $ at $ t_1 $ and $ t_2 $. Therefore, the amplitude, which is usually given as  $ \psi(\gamma,t) $ has no time dependence and is denoted by $ \psi(\gamma) $.
	
	DeWitt found that the scalar product for  $ \psi(\gamma) $  is in general only positive definite if  $ \psi(\gamma) $  does not vanish for  $ \zeta\in(-\infty,\infty) $. If one regards $ \psi(\gamma) $ as the amplitude that is the result of a path integral, then the norm of $ \psi(\gamma) $ should be positive definite. In general, the form of Eq. (\ref{wdw123}) seems to exclude  singularities in the Wheeler-DeWitt equation. This follows from a lengthy discussion in DeWitt's article \cite{DeWitt} on p. 1122-1130. 
	
	DeWitt's solution was thus to set $\psi(\gamma)=0$ at the singularity of the Friedmann cosmos.
	However, he also found that the scalar product for  $ \psi(\gamma) $  is in general only  Fortunately, DeWitt found the Friedmann cosmos to be a special case where the norm for $\psi(\gamma)$ is positive, even if one makes  $ \psi(\gamma) $  vanish at the singularity.

	At first, this idea would seem like a "resolution" of the problem. Nevertheless, the solution to set the "wave function", or the amplitude to zero at the singularity implies, if one regards  $ \psi(\gamma) $  as a path integral, that there are simply no paths arriving at the points  $ x $  in the space-time where the metric is having the singularity. Thus, one basically discusses a Friedmann universe where the singular point was artificially removed from the manifold and that no paths of the metric, or of matter amplitudes can reach it.  This is perfectly in line with our statement above, that one can not easily define path integrals over metrics with singularities in them.
\end{remark}
\begin{remark}One may argue that this just forbids singularities of the 3 metrics $\gamma_{ij}$. However, if one has a cobordism with a singular 4 metric, one may choose an appropriate time slice and investigate the path integral at the spatial hyper-surface where the singularity occurred. This often leads to a singularity in the 3 metrics.
\end{remark}Unlike for quantum field theories on Minkowski space, for path integrals over curved metrics one can not always construct a Lorentzian metric from an arbitrary Euclidean one\cite{Vis2,Vis3}. Nevertheless, Polchinski works in his book \cite{Polchinski} on string theory with a Lorentzian target space and uses a path integral over Euclidean world-sheets. This could be based on the reasoning that at least for the classical string theory, the world-sheet is just an auxiliary variable which drops out of the equations of motion. Whether the choice of the world-sheet signature is still arbitrary in the quantum theory is difficult to answer. We will now state the following
\begin{theorem}
	Using the formalism of Schulman, Laidlaw and DeWitt-Morette, the path integral $ \int D X^\mu D\tilde\gamma_{uw} e^{iS_p} $, with $ S_p $ as Polyakov action can only be defined over Lorentzian world-sheets without closed time-like curves and with metrics $ \tilde{\gamma}_{uw} $ if all the $ \tilde\gamma_{uw} $ in the paths to not contain topologically distinct space-like hyper-surfaces if no metric for the  $ \tilde\gamma_{uw} $  with singularities exists.
\end{theorem}
\begin{proof}We may cut the path integral from given times  $ t_a<t_b $  into several time slices  $ t_a<t_1<t_2<\ldots t_n<t_b $. With respect to the world-sheet metric  $ \tilde\gamma_{uw} $, we then have one-dimensional space-like hyper-surfaces at each $ t_i $ and  $ t_i+1 $ with  $ i $  as an index set. They would then serve as boundaries for the section of the world-sheet metric  $ \tilde\gamma_{uw}^{t_i,t_{i+1}} $  that describes a cobordism between $ t_i $ and  $ t_i+1 $. The string theory amplitude is a path integral over all world-sheets, including topologically distinct ones. This means that the $ t_i,t_i+1 $ can be chosen such that they are at time-slices where the space-like hyper-surfaces of $ \tilde\gamma_{uw}^{t_i,t_{i+1}} $ have different topology. If such a world-sheet does not have closed time-like curves,   $ \tilde\gamma_{uw}^{t_i,t_{i+1}} $  would need to have singularities by Geroch's theorem, which also applies for  $ 2\times 2 $  metrics. We have shown before in theorem \ref{thmm1} that we can not integrate a path integral over such metrics with the formalism of DeWitt-Morette and Cartier if there exists no norm for singular metric tensors.
\end{proof}
\begin{remark}The argument above gives a strong indication that one can not sum the path integral over metrics with singularities.
\end{remark}

Thus, one may try to to change the formalism slightly. Giddings and Horowitz\cite{Gidd, Horowitz,Horowitz2} have constructed exactly something like that for metrics of general relativity by going into the first order formalism.  The metric is then given by 
\begin{equation}
	g_{\mu\nu}=e_\mu^{\;a}e_{\nu}^{\;b}\eta_{ab},
\end{equation}where $\eta$ is the Minkowski metric and
Einstein's action becomes \begin{equation}S=\frac{1}{2}\int e^a \wedge e^b R^{cd} \epsilon_{abcd}\end{equation}
with  the field equations \begin{equation}e^b\wedge R^{cd}\epsilon_{abcd}=0,\;\;e^{\left[a\right.}\wedge De^b]\end{equation}
This represents a slight extension of general relativity. 

\begin{definition}
	Giddings \cite{Gidd} has studied the problem of topology changes and found the following metric
	
	\begin{equation}
		(\delta e,\delta e)=\int d^4x (det(e)\eta^{ab}\eta_{cd}e^\mu_a e^\nu_b  \delta e^c_\mu \delta e^d_\nu)\label{tetradnorm}
	\end{equation}
	which is finite for a tetrad given by
	
	\begin{equation}e^a_\mu(\varsigma)=\varsigma e^a_\mu(1)
	\end{equation}at $\varsigma\rightarrow 0$ as well as for $\varsigma\neq 0$.
\end{definition}

\begin{remark}
	Further calculations were made by Horowitz \cite{Horowitz,Horowitz2}. With the help of a map from the manifold $M$ into the $\mathbb{R}^4$ and another another map back to $M$, Horowitz noted that one can construct the tetrads and the curvature scalar such that the curvature scalar stays finite when one approaches a singularity. For Horowitz' construction, there would be no problems with divergences in Einstein's action when diffeomorphism invariance breaks down. 
\end{remark}
\begin{remark}
	We note that the norm of Giddings is somewhat similar to that of Clarke \cite{ClarkeGrav} and DeWitt \cite{DeWittsuperspace}. However, in DeWitt's metric a square root appears in the volume factor. This may not be able to cancel divergences from the other components of the metrics in a similar way as the metric of Giddings. For example, with a $2\times2$ metric $g_{\mu\nu}=
	\left(\begin{array}{cc}
		-\lambda	& 0 \\
		0	& \lambda
	\end{array}\right)
	$, the volume element $\sqrt{-g}$ would be proportional to $\lambda$ while expressions like $g^{\mu\alpha}g^{\nu\beta}$ can be proportional to $\lambda^{-2}$. This would create a divergent $G^{\mu\nu\alpha\beta}$ for $\lambda\rightarrow 0$ in the DeWitt metric of Eq. \ref{dewittmnetric}. In Eq. \ref{dewittmnetric}, $\delta g_{\mu\nu}$ is defined as a small deviation. Hence, it can not be assumed to cancel the  divergence of  $G^{\mu\nu\alpha\beta}$ in the DeWitt metric of Eq. \ref{dewittmnetric}.
\end{remark}

\begin{remark}
	The path integral of gravity can be written entirely in the first order formalism, see \cite{Zanelli}. We can then simply use the metric of eq. (\ref{tetradnorm}) for the homotopy considerations of DeWitt-Morette and into her path integral and no problems should arise with the space of paths or the function $I(\tau)$ if the tetrads approach a singularity constructed with Horowitz' mechanism. 
\end{remark}
\begin{remark}
	In order to define the path integral over singular world-sheets, one would likely have to find a similar first order mechanism for the world-sheet.
\end{remark}

\section{About quantum field theory close to singularities}\label{s2}

In his comment \cite{DeWitt9}, DeWitt conjectured that after all, there maybe a way to reconcile quantum mechanics with the singularities emerging from topology changes with methods from string theory. Specifically, he mentioned the results of Witten about string theory on orbifolds.

Starting from a smooth manifold  $ M $  and a group  $ G $  that acts properly discontinuously on  $ M $,  $ \mathcal{O}=M/G $  is called an orbifold. This is a space which has quotient singularities.

Orbifolds in String theory have two dimensional quotient singularities and are algebraic varieties \cite{sing} which have a resolution according to a theorem of Hironaka \cite{Hironaka}. 

\begin{definition}A resolution of a variety  $ X $  is a non-singular variety  $ Y $  if there exists a proper morphism  $ f:Y\rightarrow X $  where \begin{equation}f|_{Y\setminus f^{-1}(X_{sing})}:Y\setminus f^{-1}(X_{sing})\rightarrow X\setminus X_{sing}\label{resolv}\end{equation} is an isomorphism and the singular locus is  $ X_{sing}:=\{x\in X:\textrm{x is singular}\} $. 
\end{definition}
\begin{remark}
	In general, from a single blow-up one usually gets resolutions that still have singularities, e.g. boundary singularities may emerge. A resolution can usually be constructed by applying a series of successive blow-ups until all singularities are resolved.
\end{remark}
\begin{remark}
	In their article, \cite{Witten}, Dixon, Harvey, Vafa and Witten studied so-called crepant resolutions (which are Ricci flat) of several orbifolds and the string theory amplitude was written on them. The authors then noted that a much simpler path integral prescription with  specially "twisted" boundary conditions 
	\begin{equation} X^\mu(\sigma+2\pi,t)=g X^\mu(\sigma,t),  \;\;g|\Psi\rangle=|\Psi\rangle \label{cond1},\end{equation}where $ g $ is some element of $ G $, yielded amplitudes that were equal to that of the crepant resolution.
	
	Dixon, Harvey, Vafa and Witten were careful enough that, since their result was verified for just a few examples with a blow-up, their paper only suggested that it would hold in an arbitrary  close neighborhood of the singularity.  But they conjectured that it would hold for larger classes of orbifolds than the examples where they compared Eq. (\ref{cond1}) to the amplitude on blow-up manifold. 
\end{remark}
\begin{remark}
	For orbifolds of dimensions  $ d>4 $, it is often not known whether a Ricci flat crepant resolution exists. For example, Aspinwall has given a method to study resolutions of orbifolds of the form  $ \mathbb{C}^d/G $  where  $ d\in \mathbb{N} $  and  $ G $  is abelian. However, he notes later in the article that this method only gives crepant resolutions for  $ d\leq 3 $  \cite{Aspinwallresol}.
\end{remark}
In this section, we will argue that, with slight modifications, one can indeed extend the embedding functions and the action of string theory on some simple spaces with conical singularities up to the singular locus. We will also give an argument that string theory may be consistently written on singular space-times close to the singularity, even if a fully non-singular Ricci flat resolution can not be found.

However, we will also see that difficulties arise. For cuspidal singularities, we will prove that they often have have an influence on their neighborhood, which can make a quantum field theory undefined around the singularity.

We will start the illustration with the example of the most simple orbifold, the cone, which is given by  $ \mathbb{C}/Z_2 $  and a line element
\begin{equation}ds^2=2d\tau^2+\tau^2d\sigma^2.\end{equation}

Due to a determinant of  $ 2\tau^2 $, this metric obviously has conical singularity at $ \tau=0 $. 

The Polyakov action involves the derivative of functions  $ X^\mu $  that map into the target space. Now we let some  $ X^\mu,X^\nu,X^\eta $,where  $ \mu\neq\nu\neq\eta $  map into a  cone in 3 dimensions. We can parameterize the cone as  $ X^1(\tau,\sigma)=\tau \cos{(\sigma)}, X^2(\tau,\sigma)=\tau \sin{(\sigma)},X^0(\tau,\sigma)=\tau$.

With this parameterization, the conical singularity at  $ \tau=0 $  is the left boundary point of  $  [0,\infty) $ and we observe that unlike in a cusp, the derivatives are not becoming infinite if we are approaching the tip. 

Derivatives are usually defined on open sets. However, we can simply extend that definition to compact sets. 

\begin{definition}
	Assume one has a function  $ f(\mathbf{x})\in W\subset\mathbb{R}^{\tilde{n}},\mathbf{x}\in\mathbb{R}^n, n,\tilde{n}\in\mathbb{N}_{>0} $  with components  $ f^j $  and  $  j=1,\ldots \tilde n $  and that  $ W $  is compact. 
	
	Furthermore, let  $ \tilde{\mathbf{x}}\in\mathbb{R}^n $  be such that  $ f(\tilde{\mathbf{x}})\in\partial W  $  and  $ h\in\mathbb{R}, h>0 $. 
	if  $ f^j(\tilde{x}^0,\ldots,\tilde{x}^{i}-h,\ldots, \tilde{x}^n)\notin \partial W $  and   $ f^j(\tilde{x}^0,\ldots,\tilde{x}^{i}+h,\ldots, \tilde{x}^n) $  is undefined,    
	we set  $ h\rightarrow -h $. We then define  $ \left.\frac{\partial f^j (\mathbf{x})}{\partial x^i}\right|_{x^i=\tilde{x}^i} $  by   
	
	\begin{equation}\left.\frac{\partial f^j (\mathbf{x})}{\partial x^i}\right|_{x^i=\tilde{x}^i}=\lim_{h\rightarrow 0}\frac{f^j(\tilde{x}^0,\ldots,\tilde{x}^{i}+h,\ldots, \tilde{x}^n)-f^j (\tilde{\mathbf{x}})}{h}\label{timederivative}.\end{equation} For other cases, i.e. if both  $ f^j(\tilde{x}^0,\ldots,\tilde{x}^{i}+h,\ldots, \tilde{x}^n)\notin \partial W $  and  $ f^j(\tilde{x}^0,\ldots,\tilde{x}^{i}-h,\ldots, \tilde{x}^n)\notin \partial W $, we use the usual definition of the derivative. This allows us to describe differentiable curves that start or end at boundaries.   
	
\end{definition}

\begin{remark}From Eq. \ref{poly}, we observe that one needs the determinant and the components of the word-sheet metric $\tilde\gamma_{uw}$ in the Polyakov action. If one has a world-sheet metric with a conical singularity as above, one may have to develop a tetrad formalism with an appropriate norm for the metrics and a construction for the singular world-sheets as it was proposed by Giddings \cite{Gidd} and Horowitz \cite{Horowitz} for space-time metrics. \end{remark}

\begin{remark}
	In the case of target space singularities, it may help to use the Nambu-Goto action instead of the Polyakov action:
	
	\begin{equation}
		S_{NG}=T\int d\Sigma \sqrt{(\dot{X}\cdot X^{'})^2-\dot{X}^2X^{'2}}\label{nambu}.
	\end{equation}
\end{remark}
Let us now assume that we have used the prescription of Horowitz \cite{Horowitz} to model a singular space-time metric that vanishes at isolated points.  One observes that if some components of the "background" metric $G_{\mu\nu}$ in the scalar product $\dot{X}\cdot X^{'}=G_{\mu\nu}\dot{X}^\mu X^{'\nu}$ approach zero, the action in Eq. (\ref{nambu}) does not become divergent, as long as the derivatives of the $X^\mu$ stay finite, and in fact, Horowitz' construction ensures this too. Because no inverse of the metric appears in Eq. (\ref{nambu}), the path integral over this action is well defined. With the extension of the derivatives to compact sets, the action can even be continued up to the singular point. 

\begin{remark}
	For the cone, the singular locus will, however, not contribute to the amplitude. The action contains a Lebesgue integral, for which 
	\begin{equation} \int_{[a,b]}f(x)dx=\int_{(a,b)}f(x)dx\end{equation}
	for any function $f(x)$ if the integral on the left hand side exists. Thus the path integral will be the same as without the singular locus.
\end{remark}

\begin{remark}
	Before one can claim to successfully write the amplitude at the singularity, one has to check that e.g.  the energy of the quantum field, which is described by energy momentum tensor $T_{\mu\nu}$ that can be computed from the embedding functions $ X^\mu(\sigma,\tau) $, does not diverge at the conical singularity after quantization. Such effects can be introduced in quantum mechanics because the quantum mechanical system obeys a sort of Heisenberg's uncertainty principle and can not be localized exactly on a point. However, this would have more to do with the overall topology of the space-time and not with the singularity. 
\end{remark}
\begin{remark}
	It is also clear that the world-sheet can not be such that the string propagates over the singularity in any way. A model that would let the string e.g. not to revolve around the cone but be somewhere embedded on the surface of the cone, then propagate to the tip in time  and then over the tip, would make the spatial derivatives of the embedding functions, and thus the action, which is the integrand of the path integral, ill defined.
	
	Some notion of a derivative of the embedding functions in the action simply has to exist if we compute the path integral over them. If the action of a field theory contains derivatives of the field, singularities therefore can merely be added as boundary points of paths of quantum fields, as long as the action has no other fields that forbid this, and as long as the derivatives do not diverge close to the singularity and as long as all physical observables near and, if possible, at the singularity, remain finite classically and after quantization.
\end{remark}
\begin{remark}
	This is similar as in in classical mechanics, where singularities are sometimes defined (even though this is not an entirely satisfactory definition, see the article of Geroch \cite{Geroch123}) as points over which geodesics parameterized by some curve parameter can not be continued.
\end{remark}
\begin{remark}
	The requirement that differentiable structures must exist for quantum fields are especially visible in the axiomatic approach to quantum field theory, where field operators are described as operator valued distributions of smooth, fast falling test functions on a Schwartz space. 
	A good overview of the axiomatic methods of constructive quantum quantum field theory, especially in connection with conformal field theory and string theory, is given by Schottenloher \cite{schottenloher}.
\end{remark}

Classically, the equations of motion that follow from the Polyakov action are given by the Laplace equation \begin{equation}\nabla^2 X^\mu(\sigma,t)=\frac{1}{\sqrt{h}}\partial_u(\sqrt{\tilde\gamma} \tilde\gamma^{uw}\partial_w X^\mu(\sigma,\tau))=0\label{stringbelt}\end{equation} 
with  $ \nabla^2 $  as the Laplace-Beltrami operator.

For the case where  $ \tilde\gamma_{uw} $  is flat, (i.e. the case where the perturbation expansion of the world-sheet in the Polyakov action has been stopped before the first loop), one can use reparameterization and Weyl invariance to to set the world sheet to  $ \tilde\gamma_{00}=-1,\tilde\gamma_{11}=1,h_{01}=\tilde\gamma_{10}=0 $  in order to reduce the equations of motions to 
\begin{equation}
	\partial_{u}\partial^u X^\mu=0.
\end{equation} 

When we have such differential equations, together with a singularity that may act as a boundary point, then for the mathematician, it is clear that one has to study Cauchy problems, i.e. if, some initial data of  $ X^\mu $, or  $ \partial X^\mu $  for a set of points  $ \sigma\in[0,2\pi] $  at  $ \tau=0 $  or similarly for  $ \sigma=0 $  and  $ \tau\in[0,\infty) $, will determine a unique solution of Eq.  (\ref{stringbelt}). Let us emphasize in the following for the physicist and the mathematician who is not familiar with the study of fields in Lorentzian non-Minkowski space-times, why the study of Cauchy problems is especially important if one wants to describe the propagation of  quantum fields on a curved space-time or a space-time with boundaries or other singularities.

In constructive quantum field theory, field operators are tempered distributions and a quantum field theory of the massive Klein Gordon equation 
\begin{equation}(\square +m^2)\varphi=0\end{equation}
is famously given by
\begin{equation}\Phi(\square f+m^2f)=0\label{eq:distri}\end{equation} where  $ f $  is a fast falling smooth test function in a Schwartz space and  $ \Phi $  is an operator valued tempered distribution mapping on some Hilbert space, which also has to fulfill the Wightman axioms \cite{Wightman,schottenloher}.

Practically, one moves from a solution of the classical field to the quantum field by identifying a scalar product  $ \langle,\rangle $  for solutions  of the field equation. According to DeWitt's famous review \cite{DeWittcurved}, this works as follows: A field equation for a field $\varphi$ is given by a differential operator  $ F\varphi=0 $, where  $ F $  is self adjoint and fulfills \begin{equation}\int_M \psi_1^{*}(F\psi_2)d^4x=\int_M (F\psi_1)^{*}\psi_2 d^4x\end{equation}
for two arbitrary smooth complex functions  $ \psi_1,\psi_2 $  with compact support over the integration region. DeWitt argues that therefore, there exists a differential operator  $ \overleftrightarrow f^\mu  $  such that 
\begin{equation}\int_\Omega(\psi_1^{*}(F\psi_2)d^4x=\int_\Omega (F\psi_1)^{*}\psi_2) d^4x=\int_{\partial \Omega} \psi_1\overleftrightarrow f^\mu\psi_2 d\Sigma_\mu,\end{equation}
where  $ \partial \Omega $  is the smooth boundary of any compact region  $ \Omega $  of the space-time and  $ \Sigma_\mu $  is the outward directed surface element on  $ \partial\Omega $. For any two solutions  $ \varphi_1,\varphi_2 $  of   $ F\varphi=0 $,
one may then define a scalar product 

\begin{equation}\langle \varphi_1,\varphi_2\rangle=-i\int_\Sigma \varphi_1^{*}\overleftrightarrow f^\mu \varphi_2 d\Sigma_\mu\end{equation} with  $ \Sigma $  now as Cauchy data for  $ \varphi_1 $  and  $ \varphi_2 $. Then one has to find a complete set of functions $\varphi_j$, $\varphi_j^{*}$ that are orthonormal with respect to this scalar product. Finally, one can expand the solution  $ \Phi $  as  
\begin{equation}\Phi=\sum_j a_j \varphi_j+a^\dagger \varphi_j^* \;,\label{eq:fff}\end{equation} 
where  $ a_j^\dagger $  is the adjoint of  $ a_j $  and  $ \varphi^*_j $  is the complex conjugate of  $ \varphi_j $, with  $ a_j=\langle \varphi_j,\Phi\rangle $ and
\begin{equation}[a_j,a_{j'}^\dagger]=\delta_{ij},\;\;[a_j,a_{j'}]=0.\end{equation}

Eq.  (\ref{eq:fff}) holds for a discrete set for the orthonormal base. The expansion into an orthonormal base can, however, also be continuous, with a parameter  $ k $  and one then has 
\begin{equation}\Phi=\int dk (a_k \varphi_k+a^\dagger \varphi_k^*),  \label{eq:fff1}\end{equation}
with $ a_k=\langle \varphi_k,\Phi\rangle $  and
\begin{equation}[a_k,a_{k'}^\dagger]=\delta(k-k'),\;\; [a_k,a_{k'}]=0.\end{equation}

Mathematically, the  $ a_j $  and  $ a_k $  are operator valued tempered distributions that act on a Schwartz space of smooth, fast falling  test functions  $ f $, Hence, by expanding the field as in Eqs (\ref{eq:fff}) or (\ref{eq:fff1}), one creates a solution  $ \Phi $  that is itself an operator valued distribution which solves the field equation, as described with full mathematical notation on flat space-time in  \cite{schottenloher,Wightman}

In a curved space-time or in a space-time with boundaries and singularities, the situation is now more difficult than in flat space. For a field $\Phi(x)\rightarrow \mathbb{C}$, where $x\in M$ with $M$ Lorentzian but not a Minkowski space-time, there may be several versions of these expansions, e.g. $ \phi_k $ and  $\varphi_k $ depending on the reference frame that we have chosen, e.g. one may have

\begin{eqnarray}\Phi&=&\int dk (A_k \phi_k+A_k^\dagger \phi^{*}_k) \nonumber\\
	&=&\int dk (a_k \varphi_k+a_k^\dagger \varphi^{*}_k)\end{eqnarray}
with  $ \varphi_k $  being a positive frequency solution at early times and   $ \phi_k $  being a positive frequency solution at late times. The  $ a_k $  and  $ A_k $  are then related by a Bogoliubov transformation \cite{DeWittcurved,Birrell}.

The two expansions can have the effect that observers in different reference frames see different particle numbers. Quantum field theory in curved and general Lorentzian but non-Minkowski space-times was discovered by DeWitt and Brehme in \cite{DeWittpropa} and developed for the first time as a theory in \cite{DeWittbook1}. DeWitt wrote a nice survey on the topic in \cite{DeWittcurved}.

\begin{remark}
	At this point, the physicist should recognize why the study of Cauchy problems for the fields are especially important when the space-time is curved or has boundaries or other singularities. If the solution of a field equation  $ \Phi $  would not be determined by Cauchy data in a given reference frame, then one could, from an arbitrarily chosen reference frame, not predict anymore the observables, for example the particle number, that the other frames are observing. 
\end{remark}

\begin{remark}
	In string theory, the situation is more difficult than in ordinary quantum field theory. Here, one has a field $X^\mu:\Sigma\rightarrow T$, where both spaces, the target space-time $T$ and the world-sheet $\Sigma$ can be curved, and, classically, they are related to each other by the pull-back \cite{tong}
	
	\begin{equation}\tilde{\gamma}_{uw}=G_{\mu\nu}\partial_u X^\mu \partial_w X^\nu \label{pullback}. \end{equation}
	
	Let us suppose we consider the classical system on the orbifold at first and quantize it after we made a blow up of the target space singularity. If the blow-up was not flat, we would, by Eq. (\ref{pullback}) end up with a quantum field theory in curved world-sheets, that we would have to quantize with the methods of quantum field theory in curved space-times mentioned above, since several mode expansions of the field may exist.
\end{remark}

\begin{remark}
	It is clear that if one applies conditions like Eq.  (\ref{cond1}) to a system in flat space-time without boundaries, this does not introduce the reference frame dependent effects from quantum field theory in curved space-times that one may get from a blow-up close to the singularity if the resolution is not flat.
	
	Let us assume we make a full resolution of the orbifold that is flat. We have argued that some quotient singularities can act as boundary if they can be reached by the embedding functions. Then, if Dirichlet boundary conditions are used at the singularity, this can be expected to induce particle production effects similar to curved metrics in the singular target space, see \cite{DeWittcurved,Birrell} (the details of these effects will depend on the exact nature of the singularity). However, because we have used a full resolution that does not even have boundary singularities and is flat, these effects can then not be described by the sigma model on the blow-up. Instead, the model on the blown up target space would describe the situation where the singular locus was removed from a flat space-time, since the target space-time of the blow up was assumed to be flat.

	A blow-up of an orbifold, if it is not complete, may also have varying boundary singularities on its own. Boundaries in the target space-time act like a boundary in the world-sheet,  and they may be such that the Polyakov action can be written on them exactly. Different stages of the resolution may have different boundaries, creating different effects for the quantum field theory in non Minkowski space-times.
	
	In most cases, the blow-up is expected to have non-vanishing curvature. As a result, Eq.  (\ref{cond1}) can only be used if one ignores that the singularity and the curvature of the target space may induce particle production effects that differ from the quantum field theory on the blown up space. Additionally, if the blow-up is not flat, one has to check the consistency of the resulting theory.
\end{remark}

As we have seen, one of these checks involves whether the solutions of the differential equations on the target space are determined by Cauchy data close to the singularity.

A full study of the effects of singularities on quantum fields would go beyond this article. The classification of singularities in general relativity is mathematically difficult, see \cite{Geroch123}. The following results can be used if the singularities are conical and cuspidal. And they do not need a full resolution. In order to make the our proofs, we need the following preliminary:

\begin{definition}\label{deffaaa}Let  $ n\in\mathbb{N_+} $. An  $ n $  dimensional conic manifold  $ (\tilde{U},g) $  is a compact manifold with boundary  $ \partial{\tilde{U}} $  such that in a neighborhood  $ Y $  of any boundary component, there exists a boundary defining function  $ x $,  $ (x\geq0,\{x=0\}=\partial \tilde{U}, dx_{|\partial\tilde{U}}\neq 0) $  and  $ h\in C^\infty(\tilde{U}, Sym^2(T^{*}\tilde{U})) $  whose restriction  $ h|_{Y} $  is a positive definite metric and there is a metric  $ g $  in the interior  $ \tilde{U}^{\circ} $  which in  $ Y $  takes the degenerate form \begin{equation}g|_{Y}=dx^2 + x^2h\label{conicalmetric}.\end{equation}
\end{definition}One should regard the conic manifold as a manifold that resulted from a space whose conical singularities have been blown up into cylinders.  Similarly, one has the
\begin{definition}\label{deffb}A cuspidal manifold  $ (\tilde{U},g) $  is defined in the same way as a conical manifold, but instead of Eq.  (\ref{conicalmetric}), we have 
	\begin{equation} g|_{Y} =(1-k)(k-1)x^{2k-2}S+\mathcal{O}(x^{2k-1})dx^2+x^{2k}h,\end{equation} where  $ S $  is a smooth function on  $ \partial\tilde{U} $  and  $ k\in \mathbb{N}_{\geq 2} $. \end{definition}
Cuspidal manifolds arise as blow-ups from spaces with isolated cuspidal singularities. From the literature, we have the following two theorems:
\begin{theorem}(Melrose, Wunsch)Let  $ (\tilde{U}, g) $  be a conic manifold. Every  $ y\in \partial\tilde{U} $  is the endpoint of a unique geodesic.\label{conicaltheorem1}\end{theorem}
\begin{proof}See \cite{Wunsch}.\end{proof}
\begin{theorem}(Melrose, Wunsch)
	The Cauchy problem for the wave equation has a unique solution on the conic manifold \label{wunsch2}
\end{theorem}
\begin{proof}See \cite{Wunsch}.\end{proof}
\begin{theorem}(Grandjean, Grieser) Let  $ (\tilde{U},g) $  be a cuspidal manifold. If  $ S $  is constant, then every  $ y\in \partial\tilde{U} $  is the endpoint of a unique geodesic $\gamma_y(\tau)\in   (\tilde{U},g) $. The exponential map $exp_{\partial\tilde{U}}:\partial\tilde{U}\times(0,\tau_0)\rightarrow \tilde{U}	\setminus\partial\tilde{U},\;\;  (y, \tau )\mapsto  \gamma_y(\tau)$ defined with these geodesics is a smooth diffeomorphism. The extension $exp_{\partial\tilde{U}}:\partial\tilde{U}\times[0,\tau_0)\rightarrow \tilde{U}$ is continuous.	Assume  $ S $  is a Morse function. Then the exponential map  $exp_{\partial\tilde{U}}:(y, \tau )\mapsto  \gamma_y(\tau)$ is defined and subjective for some $\tau_0$ and some neighborhood $T\subset \partial\tilde{U}$. 
	
	Especially, the exponential map is not a diffeomorphism and the extension of the exponential map $exp_{\partial\tilde{U}}:\partial\tilde{U}\times[0,\tau_0)\rightarrow \tilde{U}$ to the boundary $\tau\rightarrow 0$ is not continuous. For an arc-length parameter $\varphi$, let  $\partial_\varphi^2 S<a_k$ on  $\partial\tilde{U}$ never take the value $2$ at a minimum. Then, the exponential map $exp_{\partial\tilde{U}}$ is a homomorphism for suitable $T$ and curve parameters $\tau_0$ . If  $\partial_\varphi^2S>a_k$ on  $\partial\tilde{U}$ at some minimum of $S$  on $\partial\tilde{U}$, then  the exponential map is not injective for any $\tau_0>0$, i.e. in  any neighborhood  $ (\tilde{U},g)$ of  $\partial\tilde{U} $  there are points through which at least 2 geodesics pass. 
	\label{cuspidaltheorem1}\end{theorem}
\begin{proof}See \cite{cuspidal}.\end{proof}
And we can use these results to prove:

\begin{theorem} For conical manifolds and cuspidal manifolds with $S=const$, there exists a solution  $  X^\mu $ of the field equations for the embedding functions that follows from the Polyakov action on $ (\tilde{U},g) $ which can be continued up to the boundary. This solution is determined by Cauchy data and also has a well defined quantum theory. If the boundary of the resolution fulfills periodic boundary conditions, then there are world-sheets that fulfill them too in the neighborhood of the singularity. \label{prooff} \end{theorem}
\begin{proof}If an  $ n $ -dimensional singular space has an isolated conical or  cuspidal singularity, then according to \cite{conical,cuspidal}, we may chose a subset $ W $ of the singular space that contains this singularity. We can then make a blow-up of the restriction $ W $ that leads to a resolution in form of a conical or cuspidal manifold. The latter are manifolds with boundary singularities, where the degenerate metric in the neighborhood  $ Y $  of any boundary component  $ \partial\tilde{U} $  can be written as in definitions (\ref{deffaaa}) or (\ref{deffb}). If differentiable functions exist on this blow-up then they also exist on the singular space in an arbitrarily close neighborhood of the singularity.
	
	As we have noted before, one can extend the definition of derivatives to compact sets. This allows us to describe differentiable curves that start or end at boundaries. The symmetric 2-tensor  $ h $  of the resolutions in the definitions (\ref{deffaaa}) or (\ref{deffb}) describes a metric on each boundary component. Since  $ h $  is positive definite by definition, it is non-degenerate, and one can choose an affine connection for the boundary manifold  $ (\partial\tilde{U},h) $.  We can therefore use a point  $ p\in\partial\tilde{U} $  as the starting point of a geodesic curve  $ \Lambda^\mu(y) $, where  $ y\in\left[0,2\pi\right] $  is some curve parameter and  $ \forall y: \Lambda^\mu(y)\in\{x=0\} $  with  $ x $  as the boundary generating function of  $ \tilde{U} $. 
	Now we take the  point   $ p=\Lambda^\mu(y) $  on  $ \partial\tilde{U} $  and use it as a start or end point for a geodesic  $ \gamma_p^\mu(t),t\in \mathbb{R}_{\geq0} $  that leaves the boundary into the interior  $ \tilde{U}^{\circ} $  and whose solution is determined by Cauchy data  $ \gamma_p^\mu(0) $  and   $ \left(\gamma_p^{\mu}\right)'(0) $.  
	
	In string theory, the classical solution of the equations of motion for $X^\mu$ that follows from the Polyakov action is equivalent to the solution of the equations of motion that follows from the Nambu-Goto action. The solution is given by the minimization of the world-sheet with respect to the target space \cite{tong}. In general, a minimal surface can be constructed by translating a minimal curve along another minimal curve, see \cite{Komerell} for a proof. 
	
	Since the exponential map is a smooth diffeomorphism by theorem (\ref{cuspidaltheorem1}), we can translate the geodesic (which is a minimal curve) $ \Lambda^\mu(y)$  along  $ \gamma_p^\mu(t) $ (which is another minimal curve) and get a differentiable minimal surface that we can parameterize by $ X^\mu(y,t)\equiv \gamma_y^\mu(t)$. We thereby construct the classical string theory world-sheet embedded in the blow-up manifold as a minimal surface on  $ \tilde{U} $  that we could get equivalently from the Polyakov action by deriving Eq.  (\ref{stringbelt}) from the action and solving Eq.  (\ref{stringbelt}) for  $ X^\mu $.
	
	Because of the construction from geodesics, the embedding functions  $ X^\mu $  that define the minimal surface are also determined by the Cauchy data close to the singularity.
	
	Eq.  (\ref{stringbelt}) that follows from the variation of the Polyakov action can also be derived from a Hamiltonian. We denote  the  derivative in  $ t $ as	$ \dot X^\mu $  and  $ X^{'\mu} $ is the derivative of  $ X^\mu $  in  $ y $. The Hamiltonian is then a function  $ H( X^\mu, X^{'\mu}, X^{''\mu}, \dot X^\mu,\ddot X^\mu )  $ which    can be determined from the Polyakov action \cite{Polchinski}. 
	Since the   $ X^\mu $, which we have constructed, solve Eq.  (\ref{stringbelt}), one can also form the same Hamiltonian   $ H( X^\mu, X^{'\mu}, X^{''\mu}, \dot X^\mu,\ddot X^\mu )  $   from the constructed  $ X^\mu $. Having obtained a Hamiltonian   $ H( X^\mu, X^{'\mu}, X^{''\mu}, \dot X^\mu,\ddot X^\mu )  $  , one can construct a symplectic manifold from  $X^\mu, X^{'\mu}, X^{''\mu}, \dot X^\mu,\ddot X^\mu $. A symplectic manifold is a Poisson manifold. The set of functions $X^\mu, X^{'\mu}, X^{''\mu}, \dot X^\mu,\ddot X^\mu $   therefore  has a well defined quantum theory by Kontsevitch's theorem \cite{Kont1} which holds for any Poisson manifold.
	
	If the boundary components  $ \partial \tilde{U} $  from \cite{conical, cuspidal} are periodical e.g in form of an  $ S^1 $, then there are geodesics  $ \Lambda^\mu(y)\in\partial\tilde{U}\in S^1 $. If one extends them maximally, they will parameterize  $ S^1 $  entirely. A blow-down will then lead to a world-sheet which has the same symmetry close to the singularity. In this way, the proof can be used to construct world-sheets of so-called "twisted" states.\end{proof}
\begin{remark}
	For conical manifolds and for cuspidal manifolds with $S=const$, the result is the same as in theorem (\ref{wunsch2}).
\end{remark}

\begin{theorem}
	Solutions $X^\mu$ that follow from the Polyakov action are uniquely determined by Cauchy data on spaces with conical singularities and cuspidal singularities if  $S=const$ arbitrarily close to the singularity. The quantum field theory of the solution  $X^\mu$ is well defined as well. 
\end{theorem}
\begin{proof}
	On the conical manifold and the cuspidal manifold with $S=const$, we have constructed in theorem \ref{prooff} a differentiable world-sheet parameterized by embedding functions $X^\mu$ that could be equivalently derived from a variation of the Polyakov action. This smooth world-sheet is determined by Cauchy data at the boundary of the conical or cuspidal manifold and has a well defined quantum theory. We can make a blow-down with a suitable map and get to the singular space where these properties of the $X^\mu$ also hold. At the singular point, the geodesics $\gamma^\mu_y(\tau)$ in \ref{cuspidaltheorem1} become then geodesics that start with varying tangential directions $y$ from the singularity. The differentiability of the exponential map in $y$ then allows to write derivatives of $X^\mu(y,t)$ in $y$ and t.
\end{proof}

\begin{remark}
	Artificial examples for spaces with cuspidal singularities are of course the orbifolds  $ \mathbb{C}/\mathbb{Z}_n $, which are cones with an angle  $ 2\pi/n $. Since the theorem only needs a local blow-up in the neighborhood of the singularity and not a global and full resolution of all singularities, it covers much more general spaces. The proof above also may work for some orbifolds that have cuspidal singularities. Examples for those spaces  are given in \cite{cusporbi}. They occur for example in hyperbolic orbifolds \cite{hyperbolic}. 
\end{remark}

\begin{theorem}
	For spaces with cuspidal singularities where $S$ is a Morse function, one can not define world-sheets with functions  $X^\mu$ that solve the equations of motion determined by the Polyakov action. One also can not quantize a world-sheet there with conventional axioms of quantum field theory.\label{nogo}
\end{theorem}

\begin{proof}
	For spaces with cuspidal singularities where $S$ is a Morse function, several geodesics starting at different points at the boundary may meet with different tangent vectors if $\partial_\varphi^2S>a_k$, see \cite{cuspidal}.  According to \cite{cuspidal}, there does not exist an extension of the exponential map $exp_{\partial\tilde{U}}:\partial\tilde{U}\times[0,\tau_0)\rightarrow \tilde{U}$
	up to the boundary which is continuous and a diffeomorphism. This implies that one can not translate the geodesics from the exponential map that start at the boundary along geodesics on the boundary defining manifold to form a differentiable minimal surface. The same holds then on the singular space at the blow-down. Hence, if one extends the solutions of the equations of motion of the Nambu-Goto action up to the singularity, then there are points in the neighborhood of the singularity, where the vectors $\partial_a X^\mu$ of the world-sheet embedded in the cuspidal manifold  by $X^\mu$ are not differentiable.
	
	Axiomatic quantum field theory constructs smooth solutions to the field equations, with field operators $\Phi$ as tempered distributions and smooth test functions $f$, which do not exist for this case. One may modify the axioms of quantum field theory and use an ansatz $X^\mu(\sigma,t)=\Theta_2(\sigma,t) X_1^\mu(\sigma,t)+\Theta_1(\sigma,t) X_2^\mu(\sigma,t)$, with step functions $\Theta_{1/2}(\sigma,t)$, where the step is at the singularity of the world-sheet and $\Theta_1(\sigma,t)=0 \forall \sigma,t: \Theta_2(\sigma,t)\neq0$ and vice-versa.  One may solve the equations of motion with $X_{1/2}^\mu$  separately for each side of the edge and expand $X_{1/2}^\mu$ in terms of field operators. The energy momentum tensor contains derivatives of the field. With  $\Theta_{1/2}(\sigma,t)$ in  $X^\mu$, the derivatives will become $\delta$ functions and diverge. 
	
	These divergences can not be treated by renormalization. The path integral is defined such that some paths describe classical solutions (they have in fact the highest weights in the amplitude). In the case of our cuspidal manifold with $S$ as Morse function, one has non-differentiable classical paths $X^\mu$ and evaluates the Nambu-Goto action, Eq. (\ref{nambu}), in order to compute the path integral. Eq. (\ref{nambu}) contains derivatives and diverges if they are divergent. According to definition \ref{counterterm}, a counter-term can modify a finite integrand of a path integral such that it converges, but one can still not subtract divergences. Since the integrand itself is divergent, we get non-renormalizable divergences for the string theory on the cuspidal manifold when $S$ is a Morse function. The same is the case for the singular target space in the limit of the blow-down. A similar result was calculated by Anderson and DeWitt and Manogue and Dray \cite{Topology,ManogueDray} for the wave equation if the latter was solved with discontinuous field functions that were then quantized.
\end{proof}

\begin{remark}
	The situation where $S\neq const$ corresponds to a perturbed cusp, see\cite{cuspcon}. Note that perturbations of a geometry should be expected to arise in quantum gravity.
	For other singularities that are neither exactly conical nor cuspidal, one would need a better understanding of the degeneracies of resolved metrics, see \cite{cuspcon}. 
\end{remark}
\begin{remark}
	A reader may argue that the cuspidal manifold was in Euclidean space-time. Indeed,  the analysis would have to be done for Lorentzian manifolds separately. However, as DeWitt noticed, in quantum gravity, one can, if there is no obstruction,  choose a coordinate system with $\alpha=1$ and $\beta_i=0$ in Eq. (\ref{dewittparam}) and the residual metric $\gamma_{ij}$ over which the path integral is summed is then purely Riemannian.  If we restrict us to this case, then a perturbed cuspidal singularity just in $\gamma_{ij}$ would cause difficulties for classical and quantum fields $X^\mu(\sigma,\tau)$ on the target space-time  if they are determined by the Polyakov action. 
\end{remark}

\begin{remark}
	Note also that the resolutions had boundary singularities and were therefore not 'full' resolutions. Hence, one can not make any conclusions about particle numbers that an observer will see with these methods. But as we have argued, it  is generally difficult to make such arguments with blow-ups due to their arbitrary curvature. Even if we had used a full resolution without boundary singularities, it may not have captured particle production effects that may be due to the singularity acting as a boundary. 
	
	The proofs therefore only show that the classical field $X^\mu$  is determined by Cauchy data at e.g. a conical singularity. Note that this then also holds for the quantum field. In the axiomatic framework of quantum field theory, the field operator $\Phi$ in Eq. (\ref{eq:fff1}) is an expansion of the solutions of the classical equations of motion in terms of a complete orthonormal base. The field operator $\Phi$ is thus also determined by Cauchy data at the singularity.
\end{remark}

\begin{remark}
	The reader may argue that one can simply restrict the classes of paths to space-times with just conical or unperturbed cuspidal singularities. This is of course true. But then one would need to define a mechanism that prevents perturbed cones if one wants to have topology changes by Geroch's theorem.
\end{remark}

\section{How the description of topology changes does not work}\label{sa}
By now, we have learned some methods that show how singularities can be made compatible with quantum fields in principle. They can be used if the singularities are such that they can act as boundary points of the propagation, and if this does not result in infinities in physical observables or difficulties with Cauchy problems, but these are rather restrictive conditions.

Especially in the string theory literature, numerous claims about the compatibility of string theory with topology changes have been made. Unfortunately, they often turn out to be hugely exaggerated. In this section we will analyze many such claims from the literature. In the next section we will then give an example that shows how one can work with conical singularities induced by a topology change and field theories in a mathematically and physically consistent way and without changing the theory. 

In string theory, two quantum field theories are called dual if they give the same physical predictions. For a sigma model, the classical equations of motion are invariant under an involution called T duality. If applied in some direction of a coordinate system, the transformation leads to a different sigma model action and a different target space, but it yields the same classical equations of motion and the same conformal field theory with the same central charge. For two amplitudes related by this duality, one can prove Buscher rules that relate the original target space metric to the T-dual one \cite{Buscher,Buscher2}.

Some physicists have applied the Buscher rules for T-duality transformations repeatedly to various coordinate directions of the target space. The result is a "target-space" whose transition functions fail to be diffeomorphisms on some subsets and which is not metrizable and called "non-geometric" by physicists \cite{Plausch, Hulla, Hullb}.

According to the definition from \cite{sing}, such a spaces are called singular. This illustrates nicely that singularities do not always come in form of a divergent Riemannian tensor.

Apparently unaware of this, physicists have tried to describe these spaces by adding new degrees of freedom in terms of a double field theory or generalized geometry \cite{Plausch, Hulla, Hullb}. Unfortunately, the resulting doubled space-time is not a resolution. If one removes the additional degrees of freedom with a constraint, one is left with the singular theory. If one keeps the additional degrees of freedom, there is no isomorphism to the singular space-time away from the singular locus. Therefore, these theories contain solutions which seem to be at odds with established physics. For example, one obtains black-holes with negative mass \cite{blackholedft}. This is at variance with the positive mass theorem of general relativity \cite{positivemass} and with usual notions of thermodynamics where the entropy of a closed system must be positive.

Physicists also have tried to apply repeated applications of T duality transformations in various directions on backgrounds of the conformal field theory (CFT) that results from string theory \cite{Blumenhagen, ng1,ng2,ng3,ng4,ng5,ng6,ng7, Plausch}. If it would be possible to define a quantum field theory on spaces like the ones in  \cite{Hulla, Hullb, Plausch}, which are non-metrizable and contain connected sets of singularities, this would be very interesting for the investigation of Lorentzian topology changes. One could then probably adopt these methods easily at the isolated singularity.

In the following, we will prove the following 
\begin{theorem}
	Non-linear Sigma models and the conformal field theories of a WZW model are incompatible with the application of repeated T-duality transformations in different directions.
\end{theorem}
\begin{proof}
	A CFT can be defined by means of path integrals  \cite{Schwarz}. If we use the rigorous definition of the path integral in \cite{functionalintegration}, the theory shows that the path measure  $ \mathcal{D}\varphi $  does only exist if the field functions  $ \varphi $  over which one integrates are elements of a metrizable Banach space or can be transformed by an affine trasformation into one. When the space of the functions  $ \varphi $  is metrizable, there must be some way to construct this metric. The only way to do so is by using the function values of the  $ \varphi $. So the functions  $ \varphi $  must map into a metrizable space themselves. Then one can create, for example, a supremum norm \begin{equation} |\varphi||_\infty=sup ||\varphi||_T,\end{equation} where  $ ||\cdot||_T $  is the norm of the target space to which  $ \varphi $  map. 
	
	The axiomatic approach to quantum field theory requires smooth tempered distributions on a Schwartz space and operators in a Hilbert space, all of which must be metrizable.
	
	However, showing a contradiction with a rigorous quantum field theory definition may be too restrictive, since there may be a broader definition of quantum field theory. 
	
	We can, however, even show that these models have internal inconsistencies. Axiomatically, a CFT is defined by the Osterwalder-Schrader axioms locality (OS1), covariance (OS2), and reflection positivity (OS3) for the correlation functions, to which the CFT axioms scaling covariance (C1), existence of an energy momentum tensor (C2) and an associative operator product expansion (C3) are added \cite{schottenloher}. 	
	
	Proponents of the non-geometric CFT's apply T-duality transformations repeatedly on the CFT given by a WZW model, which, according to their calculation would yields a non-associative theory\cite{Blumenhagen, ng1,ng2,ng3,ng4,ng5,ng6,ng7}. Thereby the CFT axiom C3 is violated. 
	By violating C3 after the application of their repeated transformation, their resulting theory is no longer a CFT. Hence, they have not applied an involution that could be called a duality, which is defined as an involution leading to a CFT with the same physical observables.
	
	The picture is then the same as with the sigma model: According to Buscher's proof, the target space-time where one applies the T duality, as well as the result of the transformation has to be a background space-time where one is able to define a nonlinear sigma model \cite{Buscher,Buscher2}. In a non-linear sigma model, the embedding functions  $ X^\mu $  with well defined derivatives  $ \partial_{a}X^\mu $  appear, see Eq.  (1) and (2) of \cite{Buscher2}. Since the embedding functions map into the target space, the existence of these derivatives requires that the target space is metrizable.
	
	Thereby, the  proof of Buscher makes it necessary that the transformations are applied on a target space which is metrizable. Furthermore, the result of this transformation should also be a metrizable space. Applying Buscher's transformation rules for the target space metric during repeated T duality transformations in different directions does not yield a target space that fulfills these requirements of Buscher's proof. The only correct conclusion would then be that one is not allowed to apply Buscher's rules in this way.
\end{proof}
\begin{remark}
	It is puzzling that the authors of \cite{Blumenhagen, ng1,ng2,ng3,ng4,ng5,ng6,ng7} do not come to this logical conclusion.
	This example can serve as a useful reminder to the practicing physicist that quantum field theories are based on aspects of functional analysis which do not only require metrizable spaces but, in general, also require differentiable structures and can not even be defined on rather "harmless" boundary singularities as long as one does not extend the notion of derivatives to compact sets as in the foregoing section \ref{s1}.
\end{remark}
\begin{remark}
	There are only slight differences in the requirements of the quantization algorithms. If considered without an expansion of the action in terms of a Taylor series or perturbation theory, the path integral can be defined for actions whose derivatives are square integrable. This is a weaker requirement than the axiomatic method of quantum field theory, which needs tempered distributions with smooth fast falling test functions in the differentiability class  $ C^\infty $. 
\end{remark}

\begin{remark}In \cite{Aspinwall,Aspinwall1}, Aspinwall, Greene, and Morrison used blow-up techniques to change the topology of string theory target spaces. They start from a space  $ X $  with singularities, which is defined to be an  $ n $  dimensional Calabi Yau hyper-surface of Fermat type in weighed projective space  $ W\mathbb{CP} $. From this space, they then construct a space  $ X/G $, where  $ G $  is the maximal subgroup of diagonal scaling symmetries on the homogeneous  $ W\mathbb{CP} $  coordinates that preserves  $ X $  and leaves the holomorphic  $ (n,0) $  form on  $ X $  invariant.
	
	Aspinwall et al. argue that  $ X $  and  $ X/G $  would be related by mirror symmetry. They show that they can use flip and flop operations to convert  $ X $  to  $ X/G $. In these operations, one applies a blow-down on a manifold, which produces a singular space. Then, one blows the singular space up to a topologically different manifold. 
	
	Finally, the authors of \cite{Aspinwall,Aspinwall1} write that one can describe the transition between  $ X $  and  $ X/G $  as a path in the K\"ahler moduli space. They claim  "we can follow paths in the complex structure moduli space which connect these complex structure limit points in a manner that encounters no physical singularity. The reason for this is that in the complex structure description, singularities arise only if the variety is not transverse." And they suggest "It would seem that a topology changing path [...] is a physically well behaved process."
	
	The authors are careful scientists and do not write that they can prove that such a process exists. For an Euclidean space-time, their statement that one can find a non-singular path between $ X $ and $ X/G $ can certainly be correct.
	
	Unfortunately, for a physical Lorentzian target space-time without closed time-like curves, this can not be true because of Geroch's theorem. One may work with Euclidean target space-times, where the topology can change without having singularities. However, if one attempts to rotate these Euclidean space-times back to the physical Lorentzian space-time by a Wick rotation, then, if there is a topology change in the resulting Lorentzian space-time, the latter will either have closed time-like curves or some sort of  singularities.  
	
	The result of Aspinwall et al. can be used to see how two different conformal field theories on Euclidean target space-times are mathematically related in the moduli space. But it can not be used to describe non-singular topology changes of Lorentzian metrics that happen without closed time-like curves.
\end{remark}
\begin{remark}
	A different approach was taken by Kiritsis and Kounnas \cite{Kir}. Similar as in the orbifold problem where string circles around the quotient singularity of a cone, these authors use configurations where the  string moves around the singularity of a topology change and does never reach it. The problem with this solution is that this is just one of the allowed physical configurations in the string theory path integral.
	
	Unfortunately, there are usually paths for the embedding functions  $ X^\mu $  where this differs. Some of them may not revolve around, but at some curve parameters,  may reach the singularity or just its neighborhood, and in some cases, this can render the entire amplitude inconsistent. One major problem to describe topology changes consistently is how to ensure that all configurations over which the path integral is summed are mathematically well defined and yield consistent and physically acceptable observables. 
\end{remark}
\section{How topology changes can be made to work for quantum fields in some cases}\label{s3}
The problems that quantum field theories have with topology changes does not only arise from singularities, but because the singularities can give rise to problematic boundary conditions of the system that are difficult to handle. Boundaries in the space-time may induce particle production from the Casimir effect\cite{DeWittcurved,Birrell}. This can, in combination with singularities, result in severe difficulties which can be studied very precisely in the trousers model of Anderson and DeWitt \cite{Topology}. 

The latter is concerned with a scalar field  $ \varphi(\sigma,t )\in\mathbb{C} $  that fulfills the 2 dimensional mass-less Klein Gordon equation  $ \square\varphi=0 $  with  $ t\in[0,\infty) $  and periodic boundary conditions in  $ \sigma $ :  $ \varphi(0,t)=\varphi(2\pi,t) $. The problem is thus similar to a closed string theory with a flat world-sheet, but that the target space is now two dimensional and given by  $ \mathbb{C} $. The two dimensional space-time, or one might say equivalently, the string world-sheet, then undergoes a topology change involving a conical singularity that increases the Euler characteristic of the one dimensional space-like hyper-surfaces of the manifold. Over the entire time interval, the topology of the space-time looks like a singular pair of pants with a conical singularity at the crotch.

Without loss of generalization, we may assume that the singularity appears at  $ \tilde\sigma,\tilde t $. If we unwrap the world-sheet as in \cite{Topology}, the topology change then appears like a wedge with two one dimensional edges  $ \tilde\sigma_1 $  and  $ \tilde\sigma_2 $  that has been cut out of a flat sheet.  The obvious problem is that  $ \varphi(x)\in\mathbb{C} $  has to fulfill different boundary conditions in the trunk \begin{equation}\varphi(0,t)=\varphi(2\pi,t)\;\forall t\in [0,\tilde t)\label{trunkeq}\end{equation}
than in the two legs, where \begin{equation}\varphi(0,t)=\varphi(\tilde{\sigma}_1,t),\;\; \varphi(\tilde{\sigma}_2,t)=\varphi(2\pi,t)\;\;\forall t>=\tilde{t}\label{legeq}\end{equation}

Anderson and DeWitt expand  $ \varphi(\sigma,t) $  into modes of in-going and outgoing states. In a subsequent article by Manogue, Copeland and Dray\cite{ManogueDray}, one can see better why an inconsistency emerges. They assume assume that field modes propagating in one leg, denoted by  $ \vartheta_L $  for the left leg and  $ \vartheta_R $  for the right leg, vanish in the other leg and vice versa. This leads to a description of the field in terms of step functions:
\begin{equation}\varphi(\sigma,t)=\Theta_1(\sigma)\varphi_L(\sigma,t)+\Theta_2(\sigma)\varphi_R(\sigma,t)
	\label{modedec}\end{equation}where 
\begin{equation}
	\Theta_1(\sigma)=\left\{\begin{aligned}1\forall\sigma<\tilde{\sigma}_1 \\0\forall  \sigma\geq\tilde{\sigma}_2
	\end{aligned}\right.,\;\; \Theta_2(\sigma)=\left\{\begin{aligned}1\forall \sigma>\tilde{\sigma}_2 \\
		0 \forall  \sigma\leq\tilde{\sigma}_1.
	\end{aligned}\right. \label{stepfunction}
\end{equation}

The energy momentum tensor in curved space-times contains derivatives of the fields. When computing the expectation values of in and out energies, Anderson, DeWitt, Copeland, Manogue and Dray find that they differ by an infinite amount even if a full renormalization method for quantum fields in curved space-times is used \cite{ManogueDray}.

However, as is known from electrodynamics, a field that hits the edges of a sharp object may just undergo diffraction processes. The author of this manuscript does think that this was not sufficiently taken into account by the authors of the trousers problem.

Usually,  first step to get a quantum field theory should be to get an orthonormal base of classical solutions (with appropriate boundary conditions) and then turn these solutions into operators.

In the case of the trousers problem, the different periodical conditions before and after the conical singularity make the finding of a solution to the equations of motion difficult. 

So we are trying to simplify the problem here, because, after all, such boundary conditions must not be there, if there is a topology change.

At the beginning of the 20.th century, Sommerfeld hat studied the diffraction of classical waves given by the Helmholtz equation when they fall upon a wedge and stated
\begin{theorem}(Sommerfeld, Malyuzhinets)
	Let  $ \theta_1 $  be the incident angle (that was restricted without loss of generality because of the symmetry)  to  $ \theta_1\in[0,\theta_w] $  of an incoming wave  $ \varphi(r,\theta)=\exp{(-ikr\cos{(\theta-\theta_1)})} $   that propagates in a flat space with coordinates  $ 0<r<\infty, -\theta_w<\theta<\theta_w $  onto a wedge parameterized by  $ 0<r<\infty, |\theta|>\theta_w $. One has the following classical solution for the exterior wedge problem:
	\begin{equation}
		\phi_0(r,\varphi)=\frac{1}{2\pi i}\int_{\gamma_+ +\gamma_-}\exp{(ikr\cos(z))}s_{1/2}(z)dz
	\end{equation}
	where  
	\begin{equation}s_1(z)=\frac{\frac{\pi}{2\theta_w}\cos{(\frac{\pi}{2\theta_w}\theta_1})}{\sin{(\frac{\pi}{2\theta_w}(\theta+z))}-\sin{(\frac{\pi}{2\theta_w}\theta_1)}}\end{equation}
	is for  Dirichlet boundary conditions  $ \phi_0(r,\theta_w)=0 $  and
	\begin{equation}
		s_2(z)=\frac{\frac{\pi}{2\theta_w}\cos{(\frac{\pi}{2\theta_w}(\theta+z))}}{\sin{(\frac{\pi}{2\theta_w}(\theta+z))}-\sin{(\frac{\pi}{2\theta_w}\theta_1)}} 
	\end{equation}
	is for Neumann boundary conditions given by  $ \frac{1}{r}\frac{\partial}{\partial \theta}\phi_0(r,\theta_w)=0 $. The  $ \gamma_+,\gamma_{-} $  are Sommerfeld contours. The  $ s_{1,2}(z) $  are meromorphic inside the domain \begin{eqnarray}\left\{-\pi-\theta_w-\epsilon_1<Re(z)<\theta_w+\epsilon_1,Im(z)>-\epsilon_2\right\}\nonumber\\\cup \left\{-\theta_w-\epsilon_1<Re(z)<\pi+\theta_w+\epsilon_1,Im(z)<\epsilon_2))\right\}\end{eqnarray}  for   $  \epsilon_{1,2}>0 $  and analytic in the same domain with  $ \epsilon_{1,2}=0^+ $. For a proof, see the review \cite{wedge} with the given references.
\end{theorem}

\begin{remark}Note that the Sommerfeld problem is described with spherical coordinates, where the tip of the cone is at the center and not with the coordinate system that is adapted to the future light-cone of the incoming wave. 
\end{remark}
We can use this to proof the following

\begin{theorem}
	Assume we have a field  $ \varphi(x^\mu)\rightarrow\mathbb{C},x^\mu\in M $  with  $ M $  as a flat manifold and that  $ \varphi $  fulfills the Klein-Gordon equation with mass  $ m $. Let the notion of derivatives of  $ \varphi $   be extended such that they hold on compact sets, as in section \ref{s2}. If topology changes of space-like hyper-surfaces of  $ M $  are induced by identifying an algebraic variety in the shape of a 3 dimensional cone as a time-like hyper-surface in  $ M $  and removing the interior of that cone out from the manifold, and if Neumann boundary conditions can be used for  $ \varphi $  at the boundary of  $ M $  where the hyper-surface was removed, and if there are no configurations in the space-time that make imposing additional boundary conditions necessary (i.e. this means that there is no need for periodic boundary conditions on  $ \varphi $  at the singularity), so that one can use asymptotic fall-of conditions for  $ \varphi $  away from the boundary, then one can expand Sommerfeld's solution of the Helmholtz equation in terms of field operators and get a quantum field theory with an energy momentum tensor  $ T_{\mu\nu}^{Neumann} $  that can be regularized to finite observables  $ \langle in |T_{\mu\nu}^{Neumann}(r,\theta,t)|out\rangle $, despite of the topology change of space-like hyper-surfaces in  $ M $.
\end{theorem}
\begin{proof}
	If one starts, e.g with a scalar field that fulfills the Klein-Gordon equation  $ (\square+m)\varphi(x,t) $, one can make a product ansatz \begin{equation}\varphi(x,t)=\varphi_0(\mathbf{x},k)\varphi(t)\end{equation}
	and arrives at the Helmholtz equation 
	\begin{equation}
		(\nabla^2+k^2)\varphi_0(\mathbf{x},k)=0
	\end{equation} 
	for the spatial parts, and for the time dependent parts we obtain an equation that can be solved by  $ \varphi(t)=\exp{i\omega t} $, where  $ \omega^2=k^2+m^2 $. 
	
	Similarly, the classical equations of motion for string theory fulfill Eq.  (\ref{stringbelt}), which is the homogeneous case of the Helmholtz equation.
	
	Geroch's theorem does not say at which metric components the singularity arises. For the Lorentzian topology change of  $ n-1 $  dimensional hyper-surfaces by an  $ n $  dimensional interpolating space-time, 
	it would not contradict Geroch's theorem if the singularity arises just in the spatial components of the  $ n\times n $  dimensional space-time metric. Let us now assume that we have a 3 dimensional flat space-time with a Lorentzian metric of signature  $ (-++) $. We want to describe this space-time in an interval from some time  $ t_1 $  to  $ t_2 $  where  $ t_2>t_1 $. We assume that there is an incoming wave-front  $ \varphi $  at some points  $ (t,u,v) $  at  $ t_1 $, with u,v as spatial coordinates. If we look at the space-like hyper-surfaces with respect to the future light-cone of the field  $ \varphi $  at time  $ t_1 $  (i.e. if we look in orthogonal directions with respect to the wave front that moves forward in time) we see a flat 2 dimensional space-like sheet. Now we want to start a process that ends with a hole in  the 2 dimensional space-like hyper-surface at time  $ t_2 $.
	
	We do this by inscribing an algebraic variety with the shape of a 3 dimensional cone into the space-time. The tip of the cone is at some point  $ (\tilde{t},r=0), t_1<\tilde t<t_2 $, where  $ \theta_w $  is the angle of the cone. Its base is at  $ t_2 $. The coordinate system of the the incoming wave front is adjusted by an incident angle  $ \theta_1 $.  We then remove the interior of the cone from the space-time, which is an open set, leaving the space-time with a boundary singularity that an incoming wave sees in its light-cone at  $ t>=\tilde{t} $. The incoming wave notices also a topology change in space-like hyper-surfaces at  $ t<\tilde{t} $  and   $ \tilde{t}>\tilde t $  with respect to the coordinate system of its future light-cone.

	For a quantum field   $ \varphi(x^\mu)\rightarrow\mathbb{C} $, a conical singularity in the space-time is analogous to a singularity in the world-sheet for the embedding functions of string theory  $ X^\mu(\sigma,\tau) $. Also, if the target space-time of a string theory has a conical singularity which can be used as a boundary point, as in the example of a cone in section \ref{s2}, then, the singularity can be regarded simply as a point with coordinates beyond which one can not continue the field and where one has to set up Dirichlet or Neumann boundary conditions for the fields  $ \varphi $  and  $ X^\mu $.
	
	The incoming wave fronts see the removed cone in their coordinate system as a usual wedge in the space-time on which they diffract.
	
	One option in the Sommerfeld problem is to choose Neumann boundary conditions. However, the fact that we have removed the space-time behind the boundary at  $ \theta<\theta_w $  means that we need to impose a condition that no momentum flows over the boundary. There is also no evidence that if we remove an open hyper-surface from a space-time, the fields would then scatter into the reverse spatial direction at the boundary. Hence the adoption of Neumann boundary conditions is useful, as long as one has not measured any sort of back scattering at a singularity. 
	
	In contrast to the assumptions of Anderson, DeWitt, Dray and Manogue from \cite{Topology,ManogueDray} one here expects that the field falls of exponentially at large distances from the wedge. So the solution can not be used for the trousers problem, but that is just because it does not have these boundary conditions.

	One now has to expand the field  $ \varphi $  as a sum of operators as in Eq.  (\ref{eq:fff1}), which is always possible since the basis functions form a complete set with respect to the scalar product  $ \langle,\rangle $  that can be defined from the field equations, and in this case is known, see \cite{Birrell} and section \ref{s2}.
	
	The energy momentum tensor  $ T_{\mu\nu} $  contains the derivatives of the field. At  $ r=\infty $   for  $ \theta>\theta_w $  they vanish because of the  $ r $  dependence of the exponential function. Similarly, with Neumann boundary conditions, the derivatives of  $ \varphi $  and thus the components of  $ T_{\mu\nu} $  vanish at the  $ t\geq \tilde{t} $,  $ \theta=\theta_w $  and one gets for the corresponding observable \begin{equation}\langle in |T_{\mu\nu}^{Neumann}(r,\theta,t)|out\rangle=finite\;\; \forall r>0,\theta>\theta_w,t\in\mathbb{R}.
	\end{equation}
	
	The calculation is a bit be more difficult if we assume that the boundary can reflect particles backwards into space. Then one would have to make use of the usual regularization methods for the Energy stress tensor in curved space-times. 
	From these methods, it can be found that the energy stress tensor has a trace anomaly where it depends on the Euler characteristic. However, one should note that we have just changed the Euler number of a space-like hyper-surfaces and not for the interpolating Lorentzian  $ 3\times3 $  metric that we have used in this example.
	
	As long as there are no discontinuities in   $ \varphi(r,\phi,t)) $  and the quantum field is not compressed to a point at the singularity, one does not get divergences in the observables of the  energy momentum tensor  $ T_{\mu\nu} $  after a suitable regularization.
\end{proof}
\begin{remark}
	Note that our toy model differs from the situation of a quantum field propagating inside of a cone, where the field could be compressed to a point at the tip, which may create divergences in the expectation values for energy due to Heisenberg's uncertainty principle.
\end{remark}
\begin{remark}
	The finite energy result is different from the result  of Anderson, DeWitt, Manogue and Dray, who computed an infinite observable of the energy momentum tensor for their singular trouser problem. Their divergences remained even after a proper regularization scheme was taken into into account. In their model, the divergences in the energy momentum tensor remained because  the field  $ \varphi $  was expressed by step functions due to the difficult periodic boundary conditions that they had. We have not found a solution for the singular trousers problem within the framework of ordinary quantum mechanics by now.
\end{remark}
\begin{remark}
	Also note that the imposition of Neumann boundary conditions at the conical singularity, which simplified the calculation, was an arbitrary assumption that probably would have to be measured by experiments. At least the author does not know of a physical consistency principle that would determine the boundary conditions at a conical singularity. After a renormalization scheme, even Dirichlet boundary should work too.
\end{remark}
\begin{remark}
	Finally, one should note that in this section, we only have described what happens to a field on the  $ n $  dimensional space-time if  $ n-1 $  dimensional  hyper-surfaces have a topology change involving a conical singularity. We have not described topology changes that involve the  $ n $   dimensional space-time.  
\end{remark} 

\section{The ground state of quantum gravity}\label{s4}
In this section, we will review some arguments which indicate why topology changes should happen in the universe. We will argue that our arguments apply to large classes of theories of quantum gravity.

Gravity is non renormalizable and has finite amplitudes only up to the one-loop order. But this suffices for some applications, e.g. simple scattering amplitudes  \cite{DeWitt3}, or the WKB solution for the Friedmann cosmos \cite{DeWitt}.

Arguments in favor of topological transitions are given by  calculations in Euclidean quantum gravity which originated from Hawking \cite{foam}.

By Einstein's equation of motion for a space-time with cosmological constant,  $ R=4\Lambda $  and if one inserts this into the Euclidean Einstein action  $ I=-\frac{1}{16\pi}\int d^4x(\sqrt{g} R-2\Lambda ) $, one gets  $ I=-\frac{\Lambda V}{8\pi} $  and from dimensional arguments, one finds  $ V(\Lambda)=\frac{f^2}{\Lambda^2} $, with $f$ as a scalar factor,   that one may substitute into the action. 
Hawking recognized that the cosmological constant acts in the action like a re-scaling of the classical gravitational action.
The behavior of the gravitational amplitude under a re-scaling was computed with Zeta function renormalization by Gibbons, Hawking and Perry from the trace anomaly \cite{Hawkingzetarenorm}. Merely, the amplitude gets multiplied by a certain factor.  If one uses this scaling factor and computes  the rest of the amplitude with a Taylor series around a classical background, one gets, provided  one only retains the first (classical) term of the Taylor series in the effective action, the following amplitude (note that we included corrections from Christensen and Duff, \cite{chris}):

\begin{equation}
	Z(\Lambda,\chi)=\left(\frac{\Lambda}{2\pi\mu^2}\right)^{-\frac{106}{90}\chi+\frac{87}{240\pi^2}f^2}\exp{\left(\frac{f^2}{8\pi\Lambda }\right)}\label{amplitude},
\end{equation}
In Eq.  (\ref{amplitude}),  $ \mu^2 $  is a renormalization scale and  $ \chi $  is the Euler characteristic of the space-time.

As we have argued, a path integral over metrics always describes a summation over cobordisms with a boundary, where the boundaries in this case should have been fixed 3 dimensional hyper-surfaces of the space-time at some beginning time  $ t_1 $  and some ending time  $ t_2 $. For these space-like hyper-surfaces (which constitute nothing else than the spatial space between  $ t_1 $  and  $ t_2 $ ), one would have to find GHY boundary terms from the second fundamental form. Since we keep these early and late hyper-surfaces fixed), this would then give us a constant that one could take as an energy in the Hamiltonian that would make the observables time dependent. The energy resulting from these boundaries would be the energy of the entire space-time.

In addition to the boundary terms at the end and the beginning of the observation, there may be other boundaries in the space-time that emerge dynamically from quantum processes. All these boundary terms then would describe the energy of the gravitational field.  In Euclidean quantum gravity it is common to set up a boundary singularity at the coordinate singularity of the event horizon if one describes a black-hole. This shows that also in the theory of relativity some singularities usually act as boundaries. It has been argued in \cite{thesispaper1} that several boundaries distributed in the space-time can yield a sum of boundary terms that acts like an energy density like the cosmological constant. 

In  \cite{foam},  Hawking argued that one has to integrate the resulting amplitude over all backgrounds and then used an inverse Laplace transform on Eq.  (\ref{amplitude}).

In contrast to his writings on zeta function regularization \cite{Hawkingzetarenorm},  Hawking called the factor  $ \mu^2 $  a cut-off in his article on space-time foam \cite{foam}, and set it to the Planck scale. However, if one reads Hawking's article on zeta function regularization, which was used to derive Eq. (\ref{amplitude}), it becomes clear that this method, which is based on expressing the divergent functional integral formally as a zeta function and using an analytic continuation to make the amplitude finite \cite{Elizalde}, employs no "cut-off" at all. Instead, the factor  $ \mu^2 $  is called an undetermined renormalization scale in the measure functional  $ \mu(\varphi) $  of DeWitt. This scale parameter should occur generally in amplitudes for any kind of fields on curved space-times and has to be found by measurements (In string theory, the undetermined string tension $ T $ plays the role of this undetermined constant). In dimensional regularization of amplitudes in curved space-times, the scale parameter also appears, and remains after the cut-off $\epsilon$ was sent to infinity. Furthermore, in his later conference papers on his space-time foam calculation, Hawking  reverts his position calls  $ \mu^2 $ again an undetermined regularization scale without giving it a value, see the collection in \cite{Hawkingbook}.

A look at the cosmology literature with respect to experimental physics showed that the same factor  $ \mu^2 $  is  present in  ordinary matter amplitudes of fields in curved-space-time. There, $ \mu^2 $  is sometimes set to  $ \mu\approx \sqrt{E_\gamma E_{grav}} $, see  \cite{Cutoff}. The reason for this is that  $ \mu^2 $  determines the order parameter for the energy of the experiment whose result is given by the amplitude. Eq.  (\ref{amplitude})  is a vacuum amplitude. The energy of the vacuum is usually measured experimentally by photons of energy  $ E_\gamma $  from supernovae with a wavelength of around  $ 500nm $. They couple to gravitons of roughly the energy  $ H_0 $. In a purely gravitational theory, as given by Eq.  (\ref{amplitude}), the photons would be replaced by other gravitational waves that arrive in an interferometer, and thus one would have an energy of order  \begin{equation} \mu\approx H_0\approx 10^{-61} \end{equation} 
in Planck units. That perturbative quantum gravity is in fact a theory of low energies was even remarked by Feynman, who wrote about the low binding energy of a gravitationally bound atom and the even lower energy corrections of the Lamb-shift effect that he wanted to compute from perturbative quantum gravity in his first lecture on the topic \cite{Feyn}. 

In the recent theoretical literature, when one works with gravity in a curved background, the constant  $ \mu $  is set to  $ \mu\approx \sqrt{|T_\mu^\mu|} $, where  $ T_\mu^\mu $  is the trace of the energy momentum tensor, see \cite{Shapiro}. From the Friedmann universe as a classical space-time one would get  $ \mu^2\approx H_0^2 $ if one neglects contributions from pressure and curvature. This may get some corrections from the energy of matter fields. We can thus reasonably set  
\begin{equation} \mu^2=|T_\mu^\mu|=\omega\Lambda,\label{lambdaeq}\end{equation} 
where  $ \Lambda $  is the contribution of gravity to the vacuum energy and  $ \omega\in \mathbb{R} $  is some initially undetermined numerical factor that accounts for the additional energy densities of matter.

In his further calculation from \cite{foam}, Hawking used an approximation of the classical gravitational action. In \cite{thesispaper1}, the author did Hawking's saddle point computation again, but for the choice  $ \mu^2=H_0^2 $ and without Hawking's approximation for the action, but with the corrected  amplitude from \cite{chris}.

It came later to the attention of this author that Hawking's original calculation was criticized by Christensen and Duff in \cite{chris}. They argued that Hawking's Laplace transformation of the amplitude would not converge for negative Euler characteristics  $ \chi $. However, that is just the case if one keeps  $ \mu^2 $  fixed. The inverse Laplace transform involves an integration over  $ \Lambda $. With Eq. (\ref{lambdaeq}), it would be wrong to integrate  $ \Lambda $  over all possible values and thus over all energy densities of the background, and at the same time keep  $ \mu^2 $  fixed, which should be the characteristic energy of a system with a loop expansion around a given background that is stopped at higher orders. Setting Eq. (\ref{lambdaeq}) with an undetermined factor $\omega$ into the inverse Laplace transform makes it indeed converge.

Now we make the assumption that the inverse Laplace transform converges. If one carefully calculates the saddle points for  $ \Lambda $  and  $ \chi $  with the recent modifications of Hawking's amplitude, one arrives at the result that the Euler characteristic is given by
\begin{equation}\chi=cV,\label{euler1}\end{equation} where  $ c $  is some numerical factor and  $ V $  is the 4 volume of the space-time. We emphasize that  $ c $  turns out to be negative for positive 4 volume if the saddle point is carefully calculated, see\cite{thesispaper1}. For an expanding universe, this means that one has to expect a negative  $ \chi $  with large absolute value, and that one finds that  $ |\chi| $  grows with the expansion of the universe. 

With the assumption that Hawking's inverse Laplace transform converges, another a saddle point calculation yields   \begin{equation} \Lambda=2\pi\mu^2, \; \text{ or }\; \omega=\frac{1}{2\pi}, \end{equation}
if no matter terms are added, see \cite{thesispaper1} for details of the calculation. 

One can now turn this argument in reverse order and say that the amplitude predicts that  \begin{equation} 2\pi\mu^2= \Lambda \end{equation} 
since  otherwise the inverse Laplace transform over the  amplitude does not converge for the negative  $ \chi $  with large  $ |\chi| $  that can be derived from the saddle points of the amplitude under the assumption of its convergence. 

With $ \mu=\sqrt{|T_\mu^\mu|} $  one furthermore expects a cosmological constant that is close to the entire  $ |T_\mu^\mu| $  because of the saddle point at  $ \Lambda=2\pi\mu^2 $.

In his phd thesis \cite{thesispaper1}, the author added matter terms to the amplitude and showed that their famously large zero point energies at first order do not change the saddle point for  $ \Lambda $  at all. Instead only  numerically small second order effects that depend on the Euler characteristic in the effective actions of the fields contributed to a small change of the saddle points for  $ \Lambda $  as quantum mechanical effects. These changes can be finely tuned by adding suitable fields. 

Phenomenologically, the model has other interesting properties. For example, one could use an expansion of the effective action that takes higher order terms into account and not only the classical term in the exponential function. Then one would get an amplitude of the form \begin{equation}\left(\frac{\Lambda}{2\pi\mu^2}\right)^{-\frac{106}{90}\chi+\frac{87}{240\pi^2}f^2} \exp{\left(\frac{\Lambda V}{8\pi}+c_1\int  d^4x\sqrt{g} R^2+c_2\chi+\ldots\right)}\end{equation}
after the inverse Laplace transform of the amplitude. The  $ R^2 $  term can account, via Starobinski's result \cite{Star}, for cosmological inflation. The leading terms in the exponential evaluated amplitude corresponds to the classical gravitational action. If one does not take the scaling factor  $ \propto (\frac{\Lambda}{2\pi\mu^2})^{-\frac{106}{90}\chi+\frac{87}{240\pi^2}f^2} $  of the amplitude into account, Hawking argued in \cite{Hawkingconstant} that one gets zero for the expectation value of the cosmological constant. Expressing Einstein's action as \begin{equation} I=-\frac{\Lambda V}{8\pi}=\frac{-f^2}{8\pi \Lambda} \end{equation} 
with  $ f $  as a dimensionless constant and just using an amplitude $Z=\exp{-I}$, one finds
\begin{equation}\Lambda=0\end{equation} as an asymptotic saddle point of the amplitude at one loop order. As a result, the scaling factor in Eq.  (\ref{amplitude}) which yields $\Lambda=2\pi\mu^2$, now appears as a small one-loop correction to this result. 

In\cite{Hawkingbook}, Hawking made the argument that the classical terms of the Euclidean effective action for the gravitational field  $ \exp{\left(\frac{1}{8\pi }\int d^4x\sqrt{g}  R\right)} $ 
would have the   $ S^4 $  space as a saddle point. If a Wick rotation to a Lorentzian space-time is (properly) employed, this would yield a DeSitter space as the background. Hawking gave no details of his calculation but just mentioned this as a result in a conference paper, claiming furthermore in a short sentence that the system would strive to the most symmetrical space without giving further details.  The action is a complicated functional, involving an integral of the curvature scalar which itself depends on the  $ 4\times4 $  matrix of the metric. To find a global saddle point of this functional is certainly a difficult task. We have noted that the gravitational path integral is defined as a functional integration over  $ 4\times4 $   metrics that form a cobordism between fixed  $ 4-1 $  dimensional hyper-surfaces. The latter are generally arbitrary and must be determined by measurements. The path integral can only give answers about the space-time between these two observations that contribute a boundary term to the action and fix the metrics at  $ t_1 $  and  $ t_2 $. To find a metric which is a cobordism between two measured and fixed hyper-surfaces may complicate the finding of a global saddle-point of the action.

For an expanding universe, Eq.  (\ref{euler1}) implies that one should expect topology changes of the space-time as the universe expands. One may argue that this calculation is based on Euclidean path integrals. However, it is a computation of the 'adiabatic' ground state when the 4-volume is held fixed for a short time. For such quasi-static situations, Euclidean methods should be unproblematic. Furthermore, one can make a similar computation with Lorentzian amplitudes.

We noted that the path integral of Lorentzian quantum gravity can be regarded as a solution  $ |\psi(\gamma,t)\rangle $  of a functional Schr\"odinger equation. Instead of the inverse Laplace transform and the saddle point calculation, one may then use  $ ||\psi(\gamma,t)\rangle|^2 $  and compute its the maximum for the parameters in the amplitude. The Lorentzian amplitude is written completely analogous to the Lorentzian amplitude. A calculation shows that results for the saddle points from   $ |\psi|^2 $  are essentially the same.

If one computes the corrections of matter terms for a (Lorentzian) quantum field theory in curved space-time with a metric  $ g_{\mu\nu} $, the effective action would usually be something like \begin{equation}Z=\exp{\left(i \ln{\left(\frac{c_0}{\mu^2}\right)} \left(c_1  \int\sqrt{-g} d^4x R+c_2   \int d^4x\sqrt{-g} R^2+c_3 \chi+\ldots\right)\right) }, \label{higherderiv} \end{equation}
where $c_i$ are coefficients that may depend on the mass of the particles and on constant factors.
The higher derivative corrections can not immediately be interpreted as corrections to Einstein's action. Instead, the effective action is just a way to rewrite the amplitude from which correlation functions can be computed, and one has to use insert the background metric  $ g_{\mu\nu} $  into these curvature terms. If one wants to compute Eq. \ref{higherderiv} e.g. for a Schwarzschild metric, then the higher derivative terms are just numerical values for a probability amplitude, without any dynamics. 

One can, however, make dynamical assumptions about  $ g_{\mu\nu} $. For example, one can assume that it is determined by Einstein's equation. In that case, the terms from the  effective matter action would get into the energy momentum tensor and influence  $ g_{\mu\nu} $  via Einstein's equation. This would yield corrections to the latter with higher derivatives. Or one could add the path integral of gravity to the amplitude. In that case, one would have to compute complicated matter and gravity interactions with corresponding vertexes. 

One problem of models that use higher derivatives in effective actions is that one can derive classical actions and equations of motion by the WKB approximation. The equations of motion would then contain higher derivatives. The Ostrogradski instability \cite{Ostro1,Ostro2} is a classical theorem which shows that non-degenerate classical actions with more than 2 derivatives have classical Hamiltonians which are not bounded from below. For gravity, some higher derivative actions may be degenerate, for example, Eq. (\ref{higherderiv}) if all other terms with higher derivatives vanish. The theory  then has positive definite energy. One can, and probably has, to use Ostrogradski's theorem to further restrict the underlying theories  which yield actions with higher derivatives. Note that Ostrogradski's result is not a quantum mechanical theorem. So having a quantum field theory that does not have unfriendly ghosts may not be a sufficient criteria. Important is also that the corrections to the classical equations of motion that one can derive from the effective action are consistent. 

In string theory, one uses a fixed so-called "background" metric  $ g_{\mu\nu} $  on which one computes the path integral of the embedding functions  $ X^\mu $  and the world-sheet  $ \tilde\gamma_{uw} $. In order to compute the effective action of that theory, one typically writes  the trace of the energy momentum tensor in terms of so-called beta functions

\begin{equation}
	T_{a}^a=\partial_u X^\mu \partial_w X^\nu(\beta_{\mu\nu}^g(X) \tilde{\gamma}^{uw}+\beta^B_{\mu\nu}(X)\epsilon^{uw})+\beta^\vartheta(X)R\end{equation}
where \begin{equation}\epsilon^{uw}=\left(\begin{array}{cc}
		0	& 1/det(\tilde\gamma_{uw}) \\
		-1/det(\tilde\gamma_{uw})	& 0
	\end{array}\right),\end{equation} and  $ B_{\mu\nu} $  and  $ \Theta $  are certain additional fields that have been added to the Polyakov action, see\cite{deligne}.

String theorists then usually make an expansion 
\begin{equation}
	X^\mu(\zeta)=X_0^\mu(\zeta)+\sqrt{\alpha'}Y^\mu(\zeta),
\end{equation} 
where  $ \alpha' $  is some function that makes the expansion dimensionless. They then use this to expand the background metric as
\begin{eqnarray}g_{\mu\nu}\partial X^\mu X^\nu=\alpha'( g_{\mu\nu}(X_0)\sqrt{\alpha'}\nonumber\\+g_{\mu\nu,\rho}(X_0)Y^\rho(\zeta)+\frac{\alpha'}{2}g_{\mu\nu,\rho\tau}Y^\rho(\zeta)Y^\tau(\zeta)+\ldots)\partial  Y^\mu(\zeta)\partial  Y^\nu(\zeta).
\end{eqnarray}

Conformal invariance dictates the vanishing of the beta functions in the energy momentum tensor. At one loop order, one gets for the beta functions related to the target space metric
\begin{equation}
	\beta_{\mu\nu}=0=\alpha'R_{\mu\nu}
\end{equation}
and for the 2 loop order, one has
\begin{equation}	\beta_{\mu\nu}^2=\alpha'R_{\mu\nu}+\frac{1}{2}\alpha'^2R_{\mu\lambda\rho\sigma}R_{\nu}^{\lambda\rho\sigma}=0.
\end{equation}

From these equations of motion one can then deduce the effective action of string theory with higher derivative terms.

The entire procedure looks extremely similar to the way  in which one gets the corrections of the effective action for matter fields in quantum field theory in curved space-times. There, one also starts by separating the fields into a background and a fluctuating contribution. The difference seems just to be that in string theory, one expands the fields in Taylor series around the background while in usual field theory, one expands the action around the background of the fields. 

However, both procedures do not determine the background. The conclusion of this would be that after computing the effective action, one has to make a further path integration. The integrand would then be the effective action of string theory, and the functional integral would go over all classical "background metrics"  $ g_{\mu\nu} $, similarly as in Hawking's calculation for the perturbatively evaluated amplitude of quantum gravity in \cite{foam}. 

Since Einstein's action and higher curvature terms appear in the effective action of both gravity coupled to the effective action from matter fields and string theory, some of results of this section then apply also for string theory, if one carries out an integration over all backgrounds, although the higher curvature terms would be different. One gets different coefficients in the scaling factor and of course the effective actions of the matter fields would be different in the sense that, for example the coefficients by which they depend on the Euler characteristic would differ.

\section{Bell's theorem, a rigorous analysis and its mathematical consequences}\label{s5}
In this section, we will review the mathematically rigorous analysis of Bell's theorem that was given by Nelson \cite{Ne3,Ne4,Ne5 } and Faris\cite{Far}. This allows us to deduce some of the properties which a physical theory should have that is not contradiction with experiments.

It will turn out that Bell's theorem consists in fact of two separate lemmas. The first was published in \cite{Bell} and has a severe mathematical loophole. This loophole was corrected by Bell in a second derivation \cite{Bell3} of his inequality and the argument was given in a mathematically rigorous form by Nelson in \cite{Ne3,Ne4,Ne5} and have been simplified by Faris\cite{Far}.

Unfortunately, Nelson's and Faris' works are not well known. This may partly be because Nelson published his article as a conference paper, arguing that the result would be trivial, and Faris published his calculations and proof in the appendix of a book intended for a general readership.
\begin{remark}
	The first article \cite{Bell} of Bell starts by defining two
	random variables  $ A(\boldsymbol{a},\lambda)=\pm1 $  and  $ B(\boldsymbol{b},\lambda)=\pm1 $,
	where  $ \boldsymbol{a} $  is the setting or axis at detector A,  $ \boldsymbol{b} $  is the setting 
	at detector B and  $ \lambda $  is some parameter over which one integrates. 
	
	In order to describe a theory with exact anti-correlations at the detectors, the random variables are defined to fulfil: 
	\begin{equation}
		A(\boldsymbol{b},\lambda)=-B(\boldsymbol{b},\lambda).\label{eq:belleqdsaf}
	\end{equation}
	Assuming  $ \rho(\lambda) $  to be the probability distribution of  $ \lambda $, Bell then writes the expectation
	\begin{equation}
		\mathrm{E}=\int d\lambda\rho(\lambda)A(\boldsymbol{a},\lambda)B(\boldsymbol{b},\lambda)\label{eq:sfwafsdf}
	\end{equation}and using Eq.  (\ref{eq:belleqdsaf}), Bell gets 
	\begin{equation}
		\mathrm{E}=-\int d\lambda\rho(\lambda)A(\boldsymbol{a},\lambda)A(\boldsymbol{b},\lambda).\label{eq:qdqfddxa}
	\end{equation}From this starting point, Bell then derives his inequality by a step by step calculation.
	
	However, already at this point, the model is in severe disagreement with quantum physics. In an EPRB experiment, the outcomes at A and B can be the results of spin measurements, or of position or momentum measurements. In each case, the observables for different settings at the same detector do not commute. Spin observables fulfill an angular momentum commutator 
	\begin{equation}
		[\hat{s}_{x},\hat{s}_{y}]=i\hbar \hat{s}_z
	\end{equation}
	which leads to an uncertainty relation for different axes.
	
	This means that upon measuring axis  $ \boldsymbol{a} $  at Station A, the measurement result for axis  $ \boldsymbol{b}\neq\boldsymbol{a} $  at the same detector A may be disturbed. As a result, one can not assume that Eq.  (\ref{eq:belleqdsaf}) would hold for the unobserved events for axis  $ \boldsymbol{b} $. One therefore can not insert Eq.  (\ref{eq:belleqdsaf}) into  Eq. (\ref{eq:sfwafsdf}) if one does not want to possibly violate Heisenberg's inequality by a purely local effect.
\end{remark}
\begin{remark}
	That there is in fact no locality assumption behind this first derivation of Bell's inequality can be seen when reformulating the same theorem in the form given by Faris in \cite{Far}. 
\end{remark}
Fortunately a few years later, Bell gave a different proof of his result in \cite{Bell3}. This time, the proof depended on a clear notion of locality and caused the attention of Nelson  \cite{Ne3,Ne4,Ne5} who brought the result into a rigorous form.

\begin{definition}
	We start by defining a probability space  $ (\Omega,\mathcal{F},P) $, with outcomes  $ \Omega $, sigma algebra  $ \mathcal{F}=\mathcal{P}(\Omega) $, where  $ \mathcal{P} $  denotes the power set, and  $ P $  as the probability measure. 
	The observables can be results of spin measurements so our measurable space becomes  $ (E,\mathcal{P}(E)) $  with  $ E=\left\{\uparrow,\downarrow\right\} $ 
	
	The outcomes are depending on experimentally chosen settings, measured at different points  $ x $  in space-time. Therefore, we have to use random fields  $ \varphi_\mu(x,\omega):(M,\Omega)\rightarrow E $  that depend on parameters  $ \mu $  which can be chosen by the experimenters at will.

	In order to make contact with Nelson's notation, we define the following notation for the outcomes at the detectors:  $ \left\{\sigma_{A}=\uparrow\right\}\equiv\left\{\uparrow\times E\right\} $,  $ \left\{\sigma_{A}=\downarrow\right\}\equiv\left\{\downarrow\times E\right\} $, 
	$ \left\{\sigma_{B}=\downarrow\right\}\equiv\left\{E\times \downarrow\right\} $,
	$ \left\{\sigma_{B}=\uparrow\right\}\equiv\left\{E\times \uparrow\right\} $. 
	
	The events  \begin{align}\left\{ \phi_{\boldsymbol{\mu}}(A,\omega)\otimes\phi_{\boldsymbol{\nu}}(B,\omega)\in\left\{\sigma_{A}=\uparrow\right\}\right\},\\\left\{ \phi_{\boldsymbol{\mu}}(A,\omega)\otimes\phi_{\boldsymbol{\nu}}(B,\omega)\in\left\{\sigma_{A}=\downarrow\right\}\right\}\end{align}give information about a spin up outcome detector at a point A for a setting  $ \mu $. Let  $ \mu $  be arbitrary and put all events that give information about an outcome at A into a sigma algebra  $ \mathcal{F}_A\subset\mathcal{F} $.
	
	Similarly, we put all events \begin{align}\left\{ \phi_{\boldsymbol{\mu}}(A,\omega)\otimes\phi_{\boldsymbol{\nu}}(B,\omega)\in\left\{\sigma_{B}=\uparrow\right\}\right\},\\\left\{ \phi_{\boldsymbol{\mu}}(A,\omega)\otimes\phi_{\boldsymbol{\nu}}(B,\omega)\in\left\{\sigma_{B}=\downarrow\right\}\right\}\end{align} for arbitrary axes  $ \nu $  into a sigma algebra  $ \mathcal{F}_B\subset\mathcal{F} $.
	
	Finally, we define a family of axis dependent probability measures  $ \mathrm{P}_{\phi_{A\boldsymbol{\mu}}\otimes\phi_{B\boldsymbol{\nu}}} $  as follows:
	
	\begin{equation}
		\mathrm{P}\left(\left\{ \phi_{\boldsymbol{\mu}}(A,\omega)\otimes\phi_{\boldsymbol{\nu}}(B,\omega)\in\left(\sigma_{A}\bigcap\sigma_{B}\right)\right\} \right)\equiv\mathrm{P}_{A_{\boldsymbol{\mu}}B_{\boldsymbol{\nu}}}\left(\sigma_{A}\bigcap\sigma_{B}\right).\label{eq:def123456}
	\end{equation} 
	
	The EPR experiment consists two stages. A measurement stage, where the outcomes are measured at two detectors located at space-like separated locations A and B and a preparation stage. The events at preparation stage take place in the overlap of the past light cones of A and B before any measurement is carried out.

	We put the events happening at preparation stage into a sigma algebra  $ \mathcal{F}_S $  and define the conditional probabilities \begin{equation}
		\mathrm{P}\left(\left.A\right|\mathcal{F}_{S}\right)(\omega)\equiv\mathrm{EX}\left[\left.1_{A}\right|\mathcal{F}_{S}\right](\omega)\label{eq:expectdsadf}
	\end{equation}
	with  $ 1_{A} $  as the indicator function of  $ A $.
\end{definition}
With this notation, one can make the following
\begin{definition}(Passive locality, Nelson)
	We call a theory passively local if 
	\begin{equation}
		\mathrm{P}_{A_{\boldsymbol{\mu}}B_{\boldsymbol{\nu}}}\left(\left.\sigma_{A}\bigcap\sigma_{B}\right|\mathcal{F}_{S}\right)=\mathrm{P}_{A_{\boldsymbol{\mu}}B_{\boldsymbol{\nu}}}\left(\left.\sigma_{A}\right|\mathcal{F}_{S}\right)\mathrm{P}_{A_{\boldsymbol{\mu}}B_{\boldsymbol{\nu}}}\left(\left.\sigma_{B}\right|\mathcal{F}_{S}\right),\label{eq:passiveloc}
	\end{equation}
	for every pair of axes  $ \boldsymbol{\mu} $  and  $ \boldsymbol{\nu} $. 
\end{definition}The violation of passive locality would imply that a dependence of the outcomes at A and B is not determined by events in  $ \mathcal{F}_S $.
Then one has an additional condition, one that forbids instantaneous signaling.

\begin{definition}(Active locality, Nelson) We call the random fields  $ \phi_{\boldsymbol{\mu}} $  and  $ \phi_{\boldsymbol{\mu'}} $   actively local if, whenever  $ \boldsymbol{\mu} $  and  $ \boldsymbol{\mu'} $  agree except on a region B in space-time, then  $ \phi_{\boldsymbol{\mu'}} $  and  $ \phi_{\boldsymbol{\mu'}} $  agree, except on the future cone of B. \end{definition}

This implies than an experimenter at B can not send a signal outside of the future cone of B. To see this, consider  $  \phi_{\boldsymbol{\mu}}(A,\omega)\otimes\phi_{\boldsymbol{\nu}}(B,\omega) $ 
and  $ \phi_{\boldsymbol{\mu'}}(A,\omega)\otimes\phi_{\boldsymbol{\nu}'}(B,\omega) $, where  $ A $  and  $ B $  in  $ M $  are space-like separated and  $ A $  is outside the future cone of  $ B $. Let  $ \mu=\mu' $  in A, then by active locality \begin{equation}\phi_{\mu}(A,\omega)=\phi_{\mu'}(A,\omega)\label{eq:activeloc},\end{equation} even if the experimenters chose different axes  $ \nu\neq \nu' $  at  $ B $.

By active locality, an event  $ \left\{ \phi_{\boldsymbol{\mu}}(A,\omega)\otimes\phi_{\boldsymbol{\nu}}(B,\omega)\in\left\{\sigma_{A}=\uparrow\right\}\right\}\in\mathcal{F}_A $  is  $ P_{A_{\boldsymbol{\mu}}B_{\boldsymbol{\nu'}}} $   equivalent to  $ \left\{ \phi_{\boldsymbol{\mu}}(A,\omega)\otimes\phi_{\boldsymbol{\nu'}}(B,\omega)\in\left\{\sigma_{A}=\uparrow\right\}\right\}\in\mathcal{F}_A $.

With these definitions, Bell's second theorem states:
\begin{theorem}(Bell's second theorem, Bell, Nelson, Faris). Let Eq.  (\ref{eq:1abcdefg}) and active and passive locality hold. Then the Clauser-Holt-Shimony-Horne (CHSH) inequality  \cite{HoltHorneClauserShimony} holds
	\begin{equation}
		\left|\mathrm{E}\left(\boldsymbol{\mu},\boldsymbol{\nu}\right)-\mathrm{E}\left(\boldsymbol{\mu},\boldsymbol{\nu}'\right)+\mathrm{E}\left(\boldsymbol{\mu}',\boldsymbol{\nu}\right)+\mathrm{E}\left(\boldsymbol{\mu}',\boldsymbol{\nu}'\right)\right|\leq2,\label{eq:2}
	\end{equation}where the function
	\begin{align}
		\mathrm{E}\left(\boldsymbol{\mu},\boldsymbol{\nu}\right) & \equiv\mathrm{P}_{A_{\boldsymbol{\mu}}B_{\boldsymbol{\nu}}}\left(\sigma_{A}=\uparrow\bigcap\sigma_{B}=\uparrow\right)+\mathrm{P}_{A_{\boldsymbol{\mu}}B_{\boldsymbol{\nu}}}\left(\sigma_{A}=\downarrow\bigcap\sigma_{B}=\downarrow\right)\nonumber \\
		& \quad-\mathrm{P}_{A_{\boldsymbol{\mu}}B_{\boldsymbol{\nu}}}\left(\sigma_{A}=\uparrow\bigcap\sigma_{B}=\downarrow\right)-\mathrm{P}_{A_{\boldsymbol{\mu}}B_{\boldsymbol{\nu}}}\left(\sigma_{A}=\downarrow\bigcap\sigma_{B}=\uparrow\right).\label{eq:2a}
	\end{align}is called correlation coefficient \end{theorem}

\begin{proof}
	For a proof, see the writings of Nelson \cite{Ne3,Ne4,Ne5}. Nelson proves an inequality which is a bit different than the statement above. In \cite{phd,schulz2}, the CHSH \cite{HoltHorneClauserShimony} inequality, which is a variant of Bell's inequality is proven. 
\end{proof}
In contrast to Bell's first article, Nelson makes sure not to make the assumption that expressions like
\begin{align}
	\mathrm{P}_{A_{\boldsymbol{\mu}}B_{\boldsymbol{\mu}}}\left(\sigma_{A}=\uparrow\right)  =  \mathrm{P}_{A_{\boldsymbol{\mu}}B_{\boldsymbol{\mu}}}\left(\sigma_{A}=\uparrow\bigcap\sigma_{B}=\downarrow\right)  =  \mathrm{P}_{A_{\boldsymbol{\mu}}B_{\boldsymbol{\mu}}}\left(\sigma_{B}=\downarrow\right)=\frac{1}{2},\nonumber\\
	\mathrm{P}_{A_{\boldsymbol{\mu}}B_{\boldsymbol{\mu}}}\left(\sigma_{A}=\downarrow\right) =  \mathrm{P}_{A_{\boldsymbol{\mu}}B_{\boldsymbol{\mu}}}\left(\sigma_{A}=\downarrow\bigcap\sigma_{B}=\uparrow\right) =  \mathrm{P}_{A_{\boldsymbol{\mu}}B_{\boldsymbol{\mu}}}\left(\sigma_{B}=\uparrow\right)=\frac{1}{2}
	\label{eq:1abcdefg}
\end{align} 

would always hold even if different axes were chosen at the detectors. Bell's inequality is violated for an experiment with several different axes at each detector and the exact anti correlations between events for the same axis  $ \mu $  at both detectors may be destroyed if a measurement for different axes  $ \mu'\neq\mu $  or  $ \nu'\neq\nu $  was carried out at the same time.

Nelson's proof is therefore removing a severe loop-hole from the proof in Bell's first article \cite{Bell}. Saying this, one should note that Nelson was much inspired by Bell's second article from \cite{Bell3}.

By active locality, the settings  $ \mu,\nu $  of the instruments at A and B which may be chosen later can have no influence on the outcomes of the random variables that generate events in  $ \mathcal{F}_S $. 
Using the experimentally verified relation (\ref{eq:1abcdefg}) for pairs of equally set axes  $ \mu,\mu $ at the two stations, Faris has shown in \cite{Far} that the passive locality condition would then imply that all events at the detectors are equivalent to events in  $ \mathcal{F}_S $  and this was then used by Faris to rewrite Bell's inequality from Bell's second theorem in the form of Bell's first theorem from \cite{Bell}.

This shows that Bell's theorem is similar to the statement of the  "Free-Will Theorem" from Conway and Kochen \cite{Freewill}. The latter also implies that if active locality holds and Eq.  (\ref{eq:1abcdefg}) holds in quantum mechanics for the observed events at an arbitrary axis  $ \mu $, then the outcomes at the detectors can not be  predetermined.

As a result, one has to conclude that any theory whose observables correspond to those of quantum mechanics must have probabilistic elements that can not be determined before the measurement actually happens.

\section{Attempts to formulate quantum mechanics from stochastic differential equations}\label{s6}
Since Kac discovered the Feynman-Kac formula in 1949 \cite{Kac}, it has been well known that Euclidean path integrals could be approximated as averages of stochastic processes, which are non-differentiable. Such a description seems not to be available for Lorentzian path integrals but there exist other techniques. Nelson has derived solutions of the non-relativistic single particle Schr\"odinger equation in terms of a stochastic process \cite{Nelson}. To do this, one starts as follows:
Any solution of the non-relativistic single particle Schr\"odinger Eq.  (\ref{Schroe}), can be written as \begin{equation}\psi(\mathbf{x},t)=\pm\sqrt{\rho(\mathbf{x},t)}e^{i\varphi(\mathbf{x},t)},\end{equation} and we can define 
\begin{equation}\mathbf{u}(\mathbf{x},t)=-\frac{\hbar}{2m}\nabla (ln(\rho(\mathbf{x},t))/\rho_0),\;\;\;\mathbf{v}(\mathbf{x},t)=\frac{\hbar}{m}\nabla\varphi(\mathbf{x},t),\label{nonrelfield}\end{equation} with  $ \rho_0 $  as some normalization factor. The Schr\"odinger equation is then found to be equivalent to \begin{equation}\frac{d}{d t}\mathbf{v}-(\mathbf{u}\nabla)\mathbf{u}+\nu\Delta u=\frac{1}{m}\mathbf{F}(\mathbf{x})\label{eq:mean},\end{equation} where  $ \mathbf{F}=-\nabla V (\mathbf{x}) $  and  $ \nu=\frac{\hbar}{2} $.

It has been shown by Nelson that this equation is connected to the mean of two stochastic differential equations which are usually written in the physicist's literature as Langevin equations for infinitely many sample trajectories  $ \mathbf{\sigma}_j $  with a friction coefficient  $ \beta=\pm m/\tau $  of opposite sign:
\begin{equation}m \ddot{\mathbf{\sigma}}_j\pm\frac{m}{\tau}\dot{\mathbf{\sigma}}_j=\mathbf{F}+\mathbf{F}^B_j\label{lang}\end{equation}
In (\ref{lang}),  $ |\beta| $  is assumed to be large and  $ \mathbf{F}^B_j $  is Wiener process with a Gaussian distribution. Writing Eq.  (\ref{lang}) one should note that the differential forms of the equations are really only abbreviations for the time integrals of these stochastic differential equations, since Brownian sample paths are non-differentiable. 

The index  $ j $  in Eq.  (\ref{lang}) is an index of the sample trajectory over which one has to average. Then one has to compute the mean of the two resulting equations, which is given by (\ref{eq:mean}). Nelson's article is rather concise. For more details of the derivation, the reader is referenced also to Fritsche and Haugk \cite{Fritsche} where the theory is also extended to many particle systems.  

It is argued in \cite{Fritsche} that the entire dynamics for a single particle is not given by a Markov process. If one looks at the evolution of a single trajectory, one has to ensure that it is governed by each one of the two Langevin equations with the same probability at each time step. One has to make a constant re-partitioning of the sample trajectories into two ensembles that are governed by different stochastic differential equations. This results in drastic differences from ordinary Markovian Brownian motion. 

The observables are computed with the state function that solves the non-relativistic Schr\"odinger equation. It depends on the fields  $ \mathbf{u}(\mathbf{x},t) $  and  $ \mathbf{v}(\mathbf{x},t) $  which are computed from  averages of all sample trajectories that arrive at a certain point  $ \mathbf{x} $  at time t. If a measurement device at a different point  $ \mathbf{x'}\neq \mathbf{x} $  at time  $ t'<t $  makes modifications to the system, e.g by closing a slit, then  a different set of trajectories will arrive at point  $ \mathbf{x} $  at  $ t $  than without the modification. This will lead to different observables. The theory therefore describes non-local phenomena.

In his phd thesis\cite{phd}, the author also conjectured that the stochastic derivation of the Schr\"odinger equation can be used to violate Bell's inequalities. One considers two systems
\begin{equation}m \ddot{\mathbf{\sigma}}_{1j}\pm\frac{m}{\tau}\dot{\mathbf{\sigma}}_{1j}=\mathbf{F}_{1}+\mathbf{F}^B_{1j}\label{langA}\end{equation}
and 
\begin{equation}m \ddot{\mathbf{\sigma}}_{2j}\pm\frac{m}{\tau}\dot{\mathbf{\sigma}}_{2j}=\mathbf{F}_{2}+\mathbf{F}^B_{2j}\label{langB}\end{equation}
where the subscripts 1 and 2 denote the two particle systems in space-like separated locations  $ 1 $  and  $ 2 $. If one would assume that one has always 
\begin{equation}\mathbf{F}^B_{2j}=-\mathbf{F}^B_{1j}\label{corr2}\end{equation}
for the Gaussian random force and that the exchange procedure between the Langevin equations with positive and negative sign works in each of the systems 1 and 2 simultaneously and in the way, then one would expect that one gets exact correlations between the observables in 1 and 2 if the experimenters there chose their measurement devices equal and deliberately set e.g.   $ \mathbf{F}_{1}=\mathbf{F}_{2} $. If the measurement devices are chosen differently, the correlations between the observables would cease to be exact, but because of Eq.  (\ref{corr2}) one could still measure some statistical dependence. Furthermore,since  $ \mathbf{F}^B_{1j}, \mathbf{F}^B_{2j} $  are not predetermined in time even if Eq.  (\ref{corr2}) holds, the outcomes at 1 and 2 would not be determined prior to measurement. 

There exist Bell inequalities  involving just position, momentum and energy for certain non-relativistic two  particle systems, see \cite{Bellineq1234}.  A verification that the model is indeed able to mimic these states and to violate a Bell inequality could come e.g from a computer simulation with a carefully chosen random process  for  $ \mathbf{F}^B_{1j}, \mathbf{F}^B_{2j} $. This would, however, be a bit difficult to implement because in each of the systems  $ 1 $  and  $ 2 $  one has to take care that  $ \mathbf{F}^B_{1j}, \mathbf{F}^B_{2j} $  as well as the exchange procedure for each trajectory under the  $ + $ and $- $ sign of the Langevin equation is such that the correlations of \cite{Bellineq1234} are then reproduced.  However, what becomes clear is that this method of stochastic quantization allows to separate the (local) influence of the measurement devises from the non-local intrinsic randomness of the system. 

In quantum mechanics, this is not so clear. Bohmian Mechanics is a non-local interpretation of quantum mechanics where  $ \mathbf{u} $  and  $ \mathbf{v} $  are interpreted as particle velocities. In an entangled multi-particle system, a measurement at one station, can reveal information about the outcomes at a spatially separated second station, provided that the experimenters there have chosen the same settings of their devises and the system as correlated. In Bohmian mechanics, this is (wrongly) interpreted as some kind of superluminal influence on the particle at a spatially separated location and that such an interpretation of quantum mechanics was possible was the reason for Einstein's complaint. Bell's theorem from \cite{Bell3} showed that one can not localize the observed randomness of the outcomes into a particle source predetermined by measurement. Nelson's analysis made clear that one can use random variables and stochastic processes to separate the local influence that experimenters may have on an entangled quantum system from the non-local and intrinsic randomness of the probability space that is given by the quantum system. 

It has been conjectured by the author that the reason for intrinsic non-local randomness of the particle motion in quantum mechanics could be a heath bath of entangled photons at spatially separated locations. These may be expected if the vacuum contains a gas of radiating entangled microscopic black-holes, which then may cause correlations like Eq.  (\ref{corr2}) that give rise to the Schr\"odinger equation for matter particles. 

Note the fact that this proposed version of stochastic mechanics, where the entire process for a given trajectory $\sigma_{1(2)j}$ is non-Markovian, and because it can maintain correlations between two separated systems if  $ \mathbf{F}^B_{1j}, \mathbf{F}^B_{2j} $ are carefully correlated, could resolve the paradox raised by Nelson in \cite{Ne4555}.

The random field in Eq.  (\ref{nonrelfield}) is modeled without relativistic effects and the fluctuating term in Eq.  (\ref{lang}) is assumed to be a non-covariant Gaussian distribution. Therefore, some sample trajectories have velocities greater than the speed of light in this model. As a result, a modification of the system at one point can yield different observables at space-like separated points.

In order to get rid of this, one must try to write covariant field equations in terms of Nelson's stochastic mechanics. Unfortunately, even though there are some attempts, see \cite{Dunkel}, a full theory of relativistic stochastic processes, relativistic statistical physics and relativistic thermodynamics has yet to be formulated. An attempt to quantize relativistic quantum fields was made by Guerra in \cite{guerra,Ne3,Ne4,Ne5}. 

Assume we have  a field  $ \varphi(x)\in\mathbb{C} $  that fulfills \begin{equation} (\square+m^2)\varphi=0,\end{equation}
with  $ x=(\mathbf{x},t)\in M $  and  $ M $  as a Minkowski space-time. We start by restricting  $ \mathbf{x} $  to a box with edge lengths  $ V $  and expand \begin{equation}\varphi(x,t)=\sum_{n=0}^\infty u_n(\mathbf{x}) q_n(t),\end{equation}
where  $ u_n(\mathbf{x}) $  fulfill the Helmholtz equation \begin{equation}\Delta u_n(\mathbf{x})=-k_n^2 u_n(\mathbf{x}),\;\; k_n^2\geq 0 \end{equation} and form a complete orthonormal base within  $ V $. The  $ q_n(t) $  satisfy the equations of motion of an harmonic oscillator \begin{equation} \ddot{q}_n(t)+(m^2+k^2_n) q_n(t)\end{equation}The system can then be quantized with stochastic processes, see \cite{Nelson,Fritsche}

After quantization and going to the infinite volume limit  $ V\rightarrow\infty $  in Minkowski space, Guerra arrives in \cite{guerra} at the conclusion that the quantum field theory for  $ \varphi $  can be described by a stochastic differential equation for the field
$ \varphi(\mathbf{x}) $ :
\begin{equation}d\varphi(x)=-\sqrt{-\Delta+m^2}\varphi(x)dt+dw(x), \label{fieldeq}\end{equation}where  $ w $  is a stochastic process.  The theory was extended to the electromagnetic field \cite{guerra} and the linearized gravitational field \cite{Davidson}. It was argued by Nelson in \cite{Ne3,Ne4,Far} that it has well behaved active locality properties for the observables and it may set to violate passive locality.
\section{Topology changes with random fields}\label{s7}
We have argued in section \ref{s4} that the ground state of quantum gravity is such that it can explain the smallness of the cosmological constant and that it also implies frequent topology changes in an expanding universe. We noted in section \ref{s0} that such changes would imply the emergence of singularities if closed time-like curves are to be avoided. In section \ref{s1} we have refuted some of DeWitt's claims against topology changes and argued that one can define path integrals over certain singular space-times in terms of tetrads with a suitable norm. In \ref{s2}, we have looked at the importance of Cauchy problems for consistent predictions of quantum field theories. We have seen in space-times with certain cuspidal singularities, some field theories may not be defined because the singularity prevents in its neighborhood the construction of fields whose derivatives exist everywhere.  We have taken note that these problems may arise for perturbed cusps. 

We have seen in sections \ref{s2} and \ref{s3} that there are good arguments which show that quantum theories for fields on a topology changing space-time can be made compatible with certain conical singularities if one gets finite observables for the fields. 

With the example of the trousers problem, we have also seen that topology changes, even if they just involve harmless conical singularities, can (but not must) induce boundary conditions that result in inconsistencies for the global analysis of the field.

It was shown recently by Krasnikov in \cite{Krasnikov} that one can get a finite quantization of the singular trousers problem if one deviates slightly from the usual quantum field theory axioms by making the fields continuous but non-differentiable. 

Krasnikov describes a mass-less field that fulfills the two dimensional wave equation and propagates from the legs to the trunk. To each leg, he associates a mode that is constant in the other leg. At points in the trunk where the modes of each leg come together, only one mode is a propagating wave while the other mode is constant. Since the modes are added together, one gets a propagating wave everywhere. The field functions that solve the wave-equation are then continuous and differentiable, except at the points of past incomplete in-extendible null geodesics, where the field functions are still continuous but their derivatives may have jump discontinuities, see figure 2 of \cite{Krasnikov}. As a result of this procedure, the energy momentum tensor stays finite. 

In the space-time foam picture that follows from Hawking's calculation, one would expect many topological changes because of the dependence of the Euler characteristic on the volume of the space-time. Therefore, one would expect the quantum field of a particle to cross light-like geodesics that emanate from one of the many singularities very often. It therefore appears that Krasnikov's solution of the singular trousers problem would make it necessary to derive quantum mechanics entirely from functions which are non-differentiable (provided that topology changes with complicated boundary conditions as in the trousers problem are allowed to happen in a space-time). 

Non-differentiable paths fit to the notion of stochastic quantization that was reviewed in the section \ref{s6}. The smooth description in the stochastic quantization schemes of section \ref{s6} only holds in the limit of a large friction coefficient, i.e for time intervals  $ dt>>\tau $, see p. 385 of \cite{Fritsche} or p. 1081 of \cite{Nelson}. If we keep  $ \tau $  finite,  $ \varphi(\mathbf{x},t) $  from (\ref{fieldeq}) is therefore smooth only approximately and becomes continuous but non-differentiable at intervals  $ dt\leq \tau $. 

In his article \cite{DeWitt}, DeWitt proved an inconsistency of the commutator algebra in quantum gravity as follows. One defines the operator
\begin{equation}\hat{\chi}^i=-2\hat\pi^{ij}_{,j}-\gamma_{il}(2\gamma_{jl,k}-\gamma_{jk,l})\hat{\pi}^{jk}
\end{equation} 
with the operator
$ \hat{\pi}^{ij}\equiv\frac{\delta}{i\delta\gamma_{ij}} $ that follows from the Wheeler-DeWitt equation. The  $ \frac{\delta}{\delta\gamma_{ij}} $  describe functional derivatives. I.e. \begin{equation}\hat{\pi}^{ij}\psi(\gamma,t)=\int \frac{\delta}{i\delta\gamma_{ij}}\psi(\gamma(x),t)\zeta_{ij}(x)d^3x\end{equation}
with  $ \zeta_{ij} $  as an arbitrary function of  $ x $. The comma denote partial derivatives with respect to  $ x^j $. In canonical quantum gravity, one has the following canonical commutator rule \begin{equation}[\gamma_{ij},\hat\pi^{k'l'}]=i\delta_{ij}^{k'l'}\label{commm}\end{equation}
DeWitt now computes in\cite{DeWitt}
\begin{eqnarray}[[\gamma_{ij},\hat\pi^{k'l'}],i\int \hat\chi_{k'}\delta\zeta^{k'}d^3x]\nonumber\\
	=([\gamma_{ij},\hat\pi^{k'l'}]\delta\zeta^k)_{,k}\label{comm1}\end{eqnarray} where  $ \delta\zeta^k $  is an arbitrary displacement. If one sets  $ x'\rightarrow x $  and uses Eq.  (\ref{commm}), one has \begin{equation}[\gamma_{ij},\hat\pi^{k'l'}]=6i\delta(x,x)\label{comm2}\end{equation}
and from the second line of Eq.  (\ref{comm1}):
\begin{equation}[[\gamma_{ij},\hat\pi^{k'l'}],i\int \hat\chi_{k'}\delta\zeta^{k'}d^3x]=([\gamma_{ij},\hat\pi^{k'l'}]\delta\zeta^k)_{,k}=6i(\delta(x,x)\delta\zeta^k),k\end{equation}
However, if one sets Eq.  (\ref{comm2}) into the first line of Eq.  (\ref{comm1}), one notes that the commutator vanishes, or 
\begin{equation}[(\gamma_{ij},\hat\pi^{k'l'}],i\int \hat\chi_{k'}\delta\zeta^{k'}d^3x]=0\end{equation} Since  $ \delta\zeta^k $  was arbitrary we arrive at a contradiction. The inconsistency of this algebra, which only follows if one sets  $ x\rightarrow x' $  puts the consistency of the amplitude Eq.  (\ref{eq:gavipath}) into question for situations where two space-time points are evaluated in the limit of a vanishing distance to each other. 

This does not look as a problem like the failure of renormalizability, since it really only occurs if two points at different locations are set to coincide. The operators  $ \hat{\pi}^{ij} $  must be regarded as acting on a state given by the amplitude. This action involves a functional derivative. The inconsistency that arises if one sets  $ x=x' $  in the algebra thus seems to point to a failure of the ability to compute functional derivatives of expressions like  $ \gamma_{ij}\psi $ properly.  It certainly implies that we can only use the Wheeler-DeWitt equation or the gravitational path integral approximately. 

Replacing the space-time points $x$ with a field operators that fulfill commutation rules as it is done in string theory is unlikely to resolve this difficulty. One could then perhaps define two different space-time coordinates as not co-measurable, but the problem of quantum gravity appears once $x'$ is equal to $x$. Furthermore, as tempered distributions, the field operators would require the ability to define smooth and fast-falling test functions, which would then have to map into the space-time. As we have seen in sections \ref{s2}, \ref{sa} and \ref{s3}, this is a notion that is not very much compatible with topology changing singularities.

The conversion of field theories in terms of stochastic processes that are continuous but non-differentiable has been somewhat successfully adapted to bosonic fields and even for linearized gravity. But that method relied on expansions of the field in terms of harmonic oscillators. Non-perturbative quantum gravity can not be brought in form of harmonic oscillators where this method would work. So one has to try a different method. 

Since we have shown that for fields on the space-time we can easily get into difficulties with cuspidal singularities and the Cauchy problem, the non-differentiability should at best be such that it leads to conical singularities only.

In the following, we will argue that one may be able to define a path integral for quantum gravity such that it can describe topology changes over metrics with singularities. In Regge calculus \cite{Regge}, one approximates the metric with simplicial complexes of various deficit angles, bone lengths and vertex numbers whose edges meet at common joints. The curvature tensor is then given by \begin{equation}R_{\mu\nu\alpha\beta}=\rho_p\epsilon_p U^p_{\alpha\beta}U^p_{\mu\nu},\label{regge}\end{equation} and the  curvature scalar is 
\begin{equation}R=2\rho_p\epsilon_p,\label{regge2}\end{equation}

where  $ \epsilon_p $  is the deficit angle associated with the vertex  $ m $  where the most edges meet at the joint  $ p $,  $ \rho $  is the bone density at the joint  $ p $  and  $ U^p_{\alpha\beta}=\epsilon_{\rho\sigma\lambda}U_p^\lambda $  with  $ \epsilon_{\rho\sigma\lambda} $  as anti-symmetric Levi-Civita symbol and  $ U^\lambda_p $  as a vector normalized to unity which points to the vertex  $ m $.

By summing over varying edge lengths, bone densities and deficit angles, this construction can be used to  approximate the path integral of gravity over all metrics, with the action constructed from Eq.  (\ref{regge2}). For the Polyakov action a description with Regge calculus was given by Jevicki and Ninomiya \cite{jav}.

In causal dynamical triangulation (CDT), one usually only considers metrics of fixed topology and ensures that the transition to different triangulations is such that the space-time develops in a causal manner, i.e. avoid singularities \cite{Loll}. It turns out to be a sufficient regulator that one can let the edge length of the simplicial complexes go to zero and the resulting amplitude is well defined. With this usual restriction to a fixed topology, the method of CDT must be extended in order to implement topological changes.

Macroscopically, space-time is perceived as smooth up to the length scales of current high energy physics experiments. So one should make the bone length of the simplicial complexes reasonably small. On the other hand, one should not make it too small. In general relativity, a curvature singularity is usually reached if the Kretschmann scalar \begin{equation}K=R_{\mu\nu\alpha\beta}R^{\mu\nu\alpha\beta}\end{equation} 
diverges. One should not attempt to approximate this situation with a vanishing bone-length, because reaching a singularity would then also make the curvature scalar  $ R $  divergent if we encounter a singularity. In that case, we could not define the integrand for the path integral anymore.

Instead, we can describe singularities in a simple way if we do not attempt to approximate the path integral  $ Z=\int \mathcal{D}g_{\mu\nu}e^{iS_g} $  over smooth metrics exactly, but limit the bone length down to some fixed but very small value where the space-time becomes non-differentiable. Eq.  (\ref{regge2}) also suggests that we have to fix the maximum number of joints  $ p $, the deficit angle  $ \epsilon_p $  and the bone densities  $ \rho $  at some large value in order to prevent divergent curvature scalars. 

The Euler characteristic of a simplicial complex is given by 
\begin{equation}\chi=V-E+F\end{equation} 
where  $ V $  is the number of vertexes,  $ E $  the number of edges and  $ F $  the number of faces. A change of the Euler characteristic can then be accomplished in principle by a suitable removal of vertices that are entirely surrounded by other simplices, provided we do not let the bone length go to zero.

On a larger scale, one then sees a differentiable manifold and recovers the original path integral of quantum gravity, while on a smaller scale one has simplicial complexes that have conical singularities and also lead to a non-differentiability of the fields that are defined on the space-time 

Because of the singularities of the simplicial complexes, one additionally has to regularize the measure. For this there exist procedures, e.g. \cite{Kath}.

\section{Conclusions and outlook}
In this article, we have carefully reviewed the singularity theorems that are connected with topological changes of a Lorentzian space-time. We have also outlined a rigorous definition of quantum field theory and applied this on the gravitational field. We have seen that the path integral of gravity is an integral over metrics that act as cobordisms, which solves the problem of time. We have found that Path integrals, independent of the action used, are only defined with paths that are elements of metric spaces. We have argued that this creates difficulties whenever one wants to define a path integral over Lorentzian metrics that are interpolating between topologically different manifolds. We nave noted that one may use tetrads and a suitable norm to overcome this problem.

We have argued that for classical and quantum fields on a space-time, it is possible to extend the notion of derivatives to compact sets and that this makes it possible in some cases to extend the fields to a singularity, similarly as it is done with a boundary. However, we also have shown that difficulties may arise in the neighbourhood of certain cuspidal singularities for some fields if one extends them to the singularity.

We then have taken a look at Hawking's calculation of the ground state of gravity. We have argued that with the right choice of the renormalization scale parameter, the amplitude converges and then leads to a  cosmological constant close to $H_0^2$. We have discussed how this model leads to inflation and why it implies frequent topology changes during the expansion of the universe. We noted that  some implications of the model may also hold for other theories, e.g. string theory.

We have reviewed the recent result that the singular trousers problem of Anderson and DeWitt can be solved by requiring the fields to be continuous but not differentiable at some points.

We have reviewed a rigorous analysis of Bell's theorem, from which we conclude that any theory of nature, should it reproduce quantum mechanics, must have probabilistic elements.

We have reviewed known attempts of stochastic quantization and found that these are characterized by trajectories which are continuous but not differentiable. This is exactly the requirement for the solution of the trousers problem.

We also have attempted to modify quantum gravity in this regard. We have argued  to describe the path integral of quantum gravity by Regge calculus and have proposed several conditions that should be obeyed if the amplitude is compatible with topological changes. Regge calculus is a natural tool for this, since it leads to a continuous space-time, which is not differentiable everywhere and where the singularities are always conical. The former condition is required for the solution of the trousers problem, the latter condition ensures that all fields are properly determined by Cauchy data close to the singularity.

Future developments would have to improve the methods of stochastic quantization. Unfortunately, the theory of relativistic stochastic processes is, despite some attempts, still not very developed and unsatisfactory. Neither fermionic fields have been described, nor has the stochastic process driving these equations been fully understood. 

If one accepts the consequences of Hawking's and Geroch's singularity theorems and of Hawking's space-time foam calculation, one has to find ways to make quantum fields compatible with the occurrence of singularities. The formalism of quantum field theory is based on aspects of functional analysis that need sufficiently smooth structures. It appears that in order to make quantum fields fully compatible with singularities, we have to change some of the fundamental definitions of quantum field theories.

\bibliographystyle{unsrt}
\bibliography{bibliography}

\end{document}